\documentclass[a4paper,fleqn,11pt,twoside]{article}

\usepackage[myheadings]{fullpage}
\def\authorrunning{\textit{M. R. Fellows et al.}}
\def\titlerunning{What is known about Vertex Cover Kernelization?}
\usepackage{authblk}

\usepackage[T1]{fontenc}
\usepackage[utf8]{inputenc}
\usepackage{lmodern}
\usepackage{microtype}
\usepackage{ellipsis}

\usepackage{stmaryrd}
\usepackage{shuffle}
\usepackage{pifont}

\usepackage{xstring}
\usepackage{mathrsfs}

\usepackage{cite}
\usepackage{verbatim}
\usepackage{fancyvrb}
\usepackage{listings}
\usepackage{amssymb}
\usepackage[fleqn]{amsmath}
\usepackage{amsfonts}
\usepackage{amsbsy}
\usepackage[fleqn,amsmath]{ntheorem}
\usepackage{mathtools}
\usepackage{dsfont}
\usepackage{nicefrac}
\usepackage{latexsym}
\usepackage{xcolor}
\definecolor{linkcolor}{RGB}{31,31,222}

\usepackage[ruled,linesnumbered,noline]{algorithm2e}
\usepackage[noend]{algpseudocode}

\usepackage[pdftex,colorlinks,linkcolor=linkcolor,citecolor=linkcolor,urlcolor=linkcolor,linktocpage]{hyperref}
\usepackage[nameinlink]{cleveref}
\usepackage{makeidx}
\usepackage[format=plain,labelsep=period,font=footnotesize,labelfont=bf,skip=2mm]{caption}[2004/07/16]
\usepackage[font=footnotesize]{subcaption}
\usepackage{courier}
\usepackage{amstext}
\usepackage{relsize}
\usepackage{scalerel}

\usepackage{enumerate}
\usepackage{multirow}
\usepackage{paralist}
\usepackage{enumitem}

\usepackage{blkarray}
\usepackage{booktabs}
\usepackage{tabulary}
\usepackage{array}
\usepackage{xspace}
\usepackage{marvosym}

\usepackage{graphicx}
\usepackage{tikz}
\usepackage{pgfplots}

\usetikzlibrary{arrows,shapes}
\usetikzlibrary{decorations.pathmorphing,backgrounds,positioning,fit,petri,decorations.pathreplacing,patterns}
\usetikzlibrary{automata,topaths}
\usetikzlibrary{decorations.text}
\usetikzlibrary{chains}
\usetikzlibrary{calc}

\let\leftold\left
\let\rightold\right
\renewcommand{\left}{\mathopen{}\mathclose\bgroup\leftold}
\renewcommand{\right}{\aftergroup\egroup\rightold}

\crefname{definition}{Definition}{Definitions}
\crefname{theorem}{Theorem}{Theorems}
\crefname{lemma}{Lemma}{Lemmata}
\crefname{corollary}{Corollary}{Corollaries}
\crefname{observation}{Observation}{Observations}
\crefname{fact}{Fact}{Facts}
\crefname{remark}{Remark}{Remarks}
\crefname{conjecture}{Conjecture}{Conjectures}
\crefname{proposition}{Proposition}{Propositions}
\crefname{figure}{Figure}{Figures}
\crefname{table}{Table}{Tables}
\crefname{section}{Section}{Sections}
\crefname{subsection}{Subsection}{Subsections}
\crefname{subsubsection}{Subsection}{Subsections}
\crefname{algorithm}{Algorithm}{Algorithms}
\crefname{example}{Example}{Examples}
\crefname{note}{Note}{Notes}
\crefname{claim}{Claim}{Claims}
\crefname{enumi}{}{}
\crefname{equation}{}{}
\crefname{reduction}{Reduction}{Reductions}
\crefname{subreduction}{Reduction}{Reductions}
\crefname{property}{Property}{Properties}

\theoremseparator{.}

\theorembodyfont{\itshape}
\theoremheaderfont{\normalfont\bfseries}
\newtheorem{reduction}{Reduction}

\newtheorem{claim}{Claim}
\renewtheorem*{claim*}{Claim}

\newtheorem{observation}{Observation}
\newtheorem{conjecture}{Conjecture}
\newtheorem{theorem}{Theorem}
\newtheorem{lemma}{Lemma}
\newtheorem{proposition}{Proposition}

\theorembodyfont{\normalfont}
\newtheorem{definition}{Definition}

\theoremstyle{nonumberplain}
\theoremheaderfont{\itshape}
\theorembodyfont{\normalfont}
\newtheorem{proof}{Proof}

\newcommand\qed{\hfill$\square$}

\makeatletter
\renewcommand{\paragraph}{%
  \@startsection{paragraph}{4}%
  {\z@}{1.5ex \@plus 1ex \@minus .2ex}{-1em}
  {\normalfont\normalsize\bfseries}%
}
\makeatother

\begin{document}

\title{What is Known About Vertex Cover Kernelization?\texorpdfstring{\thanks{Michael R.~Fellows, Lars Jaffke, Al\'{i}z Izabella Kir\'{a}ly and Frances A. Rosamond acknowledge support from the Bergen Research Foundation (BFS).}}}{{}}	

\author[1]{Michael R. Fellows}
\author[1]{Lars Jaffke}
\author[1]{Aliz Izabella Kir{\'a}ly}
\author[1]{Frances A. Rosamond}
\author[2]{Mathias Weller}

\affil[1]{Department of Informatics, University of Bergen, Norway}
\affil[ ]{\texttt{\{michael.fellows, lars.jaffke, aliz.kiraly, frances.rosamond\}@uib.no}}
\affil[2]{CNRS, LIGM, Universit\'{e} Paris EST, Marne-la-Vall\'{e}e, France}
\affil[ ]{\texttt{mathias.weller@u-pem.fr}}
	
\pgfdeclarelayer{background}
\pgfsetlayers{background,main}
\tikzstyle{vertex}=[circle, draw, fill=white]
\tikzstyle{small}=[inner sep=2.2pt]
\tikzstyle{smallvertex}=[vertex, small]
\newcommand{\N}{\ensuremath{\mathds{N}}}
\newcommand{\glue}[2]{\ensuremath{{#1}\oplus{#2}}}
\newcommand{\profile}[1]{\ensuremath{P_{#1}}}
\definecolor{FellowsEtAlTwo_gray}{rgb}{0.5,0.5,0.5}
\newcommand{\defeq}{\coloneqq}
\newcommand{\card}[1]{\left| #1\right|}
\newcommand{\bN}{\mathds{N}}
\newcommand{\bR}{\mathds{R}}
\newcommand{\cO}{\mathcal{O}}
\newcommand{\cX}{\mathcal{X}}
\newcommand{\vc}{\textsc{Vertex Cover}}
\newcommand{\ETH}{\textsf{ETH}}
\newcommand{\NP}{\mathsf{NP}\xspace}
\newcommand{\coNP}{\mathsf{coNP}\xspace}
\newcommand{\poly}{\mathsf{poly}\xspace}
\newcommand{\FPT}{\mathsf{FPT}\xspace}
\newcommand{\barrierconstant}{\zeta_{VC}}
\newcommand{\barrierdegree}{\delta_{VC}}
\renewcommand{\thereduction}{R.\arabic{reduction}}
\renewcommand{\thesubreduction}{R.\arabic{reduction}.\arabic{subreduction}}
\newcommand{\claimqed}{\hspace*{\fill}$\lrcorner$}
\newenvironment{clproof}{\begin{proof}}{\end{proof}}

\maketitle	

\setcounter{footnote}{0}

\begin{abstract}
	We are pleased to dedicate this survey on kernelization of the \textsc{Vertex Cover} problem, to Professor Juraj Hromkovi\v{c} on the occasion of his 60th birthday. The \textsc{Vertex Cover} problem is often referred to as the \emph{Drosophila} of parameterized complexity. It enjoys a long history. New and worthy perspectives will always be demonstrated first with concrete results here. This survey discusses several research directions in \textsc{Vertex Cover} kernelization. The \emph{Barrier Degree} of \textsc{Vertex Cover}  is discussed. We have reduction rules that kernelize vertices of small degree, including in this paper new results that reduce graphs almost to minimum degree five. Can this process go on forever? What is the minimum vertex-degree barrier for polynomial-time kernelization? Assuming the Exponential-Time Hypothesis, there is a minimum degree barrier. The idea of \emph{automated kernelization} is discussed. We here report the first experimental results of an AI-guided branching algorithm for \textsc{Vertex Cover} whose logic seems amenable for application in finding reduction rules to kernelize small-degree vertices. The survey highlights a central open problem in parameterized complexity. Happy Birthday, Juraj!
\end{abstract}
	
\section{Introduction and Preliminaries}
A vertex cover of a graph is a subset of its vertices containing at least one endpoint of each of its edges. The \textsc{Vertex Cover} problem asks, given a graph $G$ and an integer $k$, whether $G$ contains a vertex cover of size at most $k$.

The study of the \vc{} problem lies at the roots of the theory of $\NP$-completeness: It is one of Karp's 21 $\NP$-complete problems~\cite{FellowsEtAlTwo_Kar72} and plays a central role in the monograph of Garey and Johnson~\cite{FellowsEtAlTwo_GJ79}. However, interest in the \vc~problem reaches far beyond pure theory. One reason is that it naturally models \emph{conflict resolution},\footnote{In the textbook~\cite{FellowsEtAlTwo_CFK15}, the problem was entertainingly introduced as `Bar Fight Prevention'.} 
a problem occurring in numerous scientific disciplines, with an international workshop devoted to it~\cite{FellowsEtAlTwo_COREDEMA16}. Other applications include~classification methods (see, e.g., \cite{FellowsEtAlTwo_GKK14}), computational biology~(e.g., \cite{FellowsEtAlTwo_CLV08}), and various applications follow from the duality of \textsc{Vertex Cover} with the \textsc{Clique} problem (see, e.g., \cite{FellowsEtAlTwo_ACF04}). The latter finds numerous applications in fields such as computational biology and bioinformatics~\cite{FellowsEtAlTwo_BW06,FellowsEtAlTwo_Kar11,FellowsEtAlTwo_KLW96,FellowsEtAlTwo_SBS05,FellowsEtAlTwo_YSK04}, computational chemistry~\cite{FellowsEtAlTwo_DW96,FellowsEtAlTwo_LG07,FellowsEtAlTwo_WBD98}, and electrical engineering~\cite{FellowsEtAlTwo_CS93,FellowsEtAlTwo_HP98}.

In \emph{parameterized/multivariate algorithmics}~\cite{FellowsEtAlTwo_CFK15,FellowsEtAlTwo_DF13,FellowsEtAlTwo_Nie02}, the objects of study are computational problems whose instances are additionally equipped with a integer $k$, the \emph{parameter}, typically expressing some structural measure of the instance of the problem. The goal is to design algorithms for hard problems whose runtime confines the combinatorial explosion to the parameter $k$ rather than the size of the input. A parameterized problem is called \emph{fixed-parameter tractable} if it can be solved in time $f(k)\cdot n^{\cO(1)}$ where $f$ is some computable function, $k$ the parameter and $n$ the input size. The second central notion in the field of parameterized algorithms is that of a \emph{kernelization}~\cite{FellowsEtAlTwo_DF95,FellowsEtAlTwo_DFS99,FellowsEtAlTwo_Fel06}, a polynomial-time algorithm (usually described as a set of \emph{reduction rules}) that takes as input an instance $(I, k)$ of a parameterized problem and outputs an equivalent instance $(I', k')$, where $\card{I'} + k' \le g(k)$ for some computable function $g$.\footnote{As the focus of this text is on the \vc ~problem, we refer to~\cite{FellowsEtAlTwo_GN07,FellowsEtAlTwo_LMS12} for general surveys on the subject of kernelization and to~\cite{FellowsEtAlTwo_MRS11} for a survey on the corresponding lower bound machinery.}

Kernelization (for the first time!) provided a theory of preprocessing with mathematically provable guarantees. On the other end, kernelization has immediate practical implications, as demonstrated  by Karsten Weihe's problem~\cite{FellowsEtAlTwo_Wei98,FellowsEtAlTwo_Wei00} (see also~\cite{FellowsEtAlTwo_Fel00,FellowsEtAlTwo_Fel01}) concerning the train systems in Europe. By the means of two simple reduction rules, graphs (instances) on $10,000$ vertices are reduced to equivalent instances whose connected components are of size at most $50$, making the reduced instance solvable exactly even by brute force in reasonable time, after the preprocessing, even though the general problem is $\NP$-hard.
Similar reduction rules have been successfully applied in the context of cancer research~\cite{FellowsEtAlTwo_BM10} and spread of virus~\cite{FellowsEtAlTwo_EM17}.

The notions of fixed-parameter tractability and kernelization are tightly linked. It has been shown by Cai et al. that a parameterized problem is fixed-parameter tractable if and only if it has a (polynomial-time) kernelization algorithm~\cite{FellowsEtAlTwo_CCDF97}.
Kernelization for the \vc ~problem, which is often referred to as the \emph{Drosophila} of parameterized complexity~\cite{FellowsEtAlTwo_DF13,FellowsEtAlTwo_FGG10,FellowsEtAlTwo_GNW07,FellowsEtAlTwo_Nie02}, enjoys a long history. In 1993, the first kernel on $\cO(k^2)$ vertices was obtained, and is accredited to Buss~\cite{FellowsEtAlTwo_BG93}, with more refined reduction rules given in~\cite{FellowsEtAlTwo_BFR98}. Kernels with a linear number of vertices were obtained in various ways. Using classic graph theoretic results, Chor et al.~gave a kernel on $3k$ vertices~\cite{FellowsEtAlTwo_CFJ04} (see also~\cite{FellowsEtAlTwo_Fel03}), a kernel on $2k$ vertices was obtained via an LP-relaxation by Chen et al.~\cite{FellowsEtAlTwo_CKJ01} and another kernel on $2k$ vertices without the use of linear programming was obtained by Dehne et al.~\cite{FellowsEtAlTwo_DFR04}.
The next series of improvements gave kernels on $2k - c$ vertices~\cite{FellowsEtAlTwo_SY11} and the current champion which is due to Lampis has $2k-c\log k$ vertices~\cite{FellowsEtAlTwo_Lam11}, where in the latter two $c$ is any fixed constant. Another kernel on $2k - \cO(\log k)$ vertices was observed in~\cite{FellowsEtAlTwo_NRR12}. An experimental evaluation of several of the earlier kernels was carried out in~\cite{FellowsEtAlTwo_ACF04}.

There is no known subquadratic bound on the number of \emph{edges} in any kernel for \textsc{Vertex Cover}, and the question whether such a kernel exists was a long standing open question in multivariate algorithmics. It was finally shown that up to logarithmic factors, \vc ~kernels with a quadratic number of edges are likely to be optimal: Dell and van Melkebeek, building on work of Bodlaender, Downey, Fellows and Hermelin~\cite{FellowsEtAlTwo_BDFH08,FellowsEtAlTwo_BDFH09}, also Fortnow and Santhanam~\cite{FellowsEtAlTwo_FS08}, showed that there is no kernel on $\cO(n^{2 - \varepsilon})$ bits, for any $\varepsilon > 0$, unless $\NP \subseteq \coNP/\poly$~\cite{FellowsEtAlTwo_DM14}. The latter would imply that the polynomial hierarchy collapses to its third level~\cite{FellowsEtAlTwo_Yap83} which is widely considered to be implausible by complexity theorists.\footnote{We also refer to~\cite[pages 19f]{FellowsEtAlTwo_Jan13} and~\cite[Appendix A]{FellowsEtAlTwo_Wel13} for brief accounts of the implausibility of $\NP \subseteq \coNP/\poly$.}

In another line of research, following the \emph{parameter ecology} program~\cite{FellowsEtAlTwo_FJR13}, the existence of kernels for \vc ~w.r.t.\ parameters that take on smaller values than the vertex cover number was studied. Such parameterizations are typically referred to as \emph{structural} parameterizations of \vc. The first such result is due to Jansen and Bodlaender who gave a kernel on $\cO(\ell^3)$ vertices, where $\ell$ is the size of a feedback vertex set of the graph~\cite{FellowsEtAlTwo_JB13}. Further results include polynomial kernels where the parameter is the size of an odd cycle traversal or a K{\"o}nig deletion set~\cite{FellowsEtAlTwo_KW12}, the size of vertex deletion sets to maximum degree at most two~\cite{FellowsEtAlTwo_MRS15}, pseudoforest~\cite{FellowsEtAlTwo_FS16} and $d$-quasi forest~\cite{FellowsEtAlTwo_HK17}, or small treedepth~\cite{FellowsEtAlTwo_BS17}. Using the above mentioned lower bound machinery, it was shown that there is no kernel polynomial in the size of a vertex deletion set to chordal or perfect graphs unless $\NP \subseteq \coNP/\poly$~\cite{FellowsEtAlTwo_BJK14,FellowsEtAlTwo_FJR13}.


As \vc{} is the primary intellectual ``lab animal'' in parameterized complexity, \emph{new and worthy perspectives will always be demonstrated first with concrete results here.}
We discuss several research directions in (\vc) kernelization.
The first one is based on the observation that several reduction rules are known to kernelize  vertices of small degree~\cite{FellowsEtAlTwo_BG93,FellowsEtAlTwo_FS99,FellowsEtAlTwo_Stege00}; a natural question is whether this process can go on `forever', i.e., whether we can find, for any fixed constant $d \in \bN$, a set of reduction rules that kernelize in polynomial time to a reduced graph (the kernel) of minimum degree  $d$. 
On the negative side, we observe that unless the Exponential-Time Hypothesis~\cite{FellowsEtAlTwo_IP01,FellowsEtAlTwo_IPZ01} fails, this is not the case even if the exponent in the polynomial-time kernelization is some arbitrary function of  $d$.
On the positive side, we give a clear account of reduction rules for \vc ~that were first observed by Fellows and Stege~\cite{FellowsEtAlTwo_FS99} that kernelize instances to minimum degree `almost five' (see \cref{FellowsEtAlTwo_thm:kernelization:degree} for the exact statement) and discuss how this question is closely related to finding faster fpt-algorithms for \vc, a question that lies at the very heart of parameterized complexity research.

In the light of the ongoing machine-learning and artificial intelligence revolution, one might wonder whether AI could assist in the search for new reduction rules of parameterized problems as well. While this question seems far out, we report first experimental results of an AI-guided branching algorithm for \textsc{Vertex Cover} whose logic seems amenable for application in finding new reduction rules to kernelize to increasing minimum degree.

The rest of this paper is organized as follows. In the remainder of this section, we give preliminary definitions and introduce the necessary background. In \cref{FellowsEtAlTwo_sec:standard:methods} we review some classic \textsc{Vertex Cover} kernels.  \Cref{FellowsEtAlTwo_sec:small:degree} is devoted to the topic of kernelizing small-degree vertices. We there give a description of reduction rules observed by Fellows and Stege~\cite{FellowsEtAlTwo_FS99} (see also~\cite{FellowsEtAlTwo_Stege00}). In Section~\cref{FellowsEtAlTwo_sec:automated} we report results on an AI-guided branching algorithm whose ideas might lay the foundations of \emph{automatically generated} reduction rules for \textsc{Vertex Cover}.  We conclude with an  open problem in \cref{FellowsEtAlTwo_sec:open:problems}.
	
\paragraph{Technical Preliminaries and Notation.} For two integers $a$ and $b$ with $a < b$, we let $[a..b] \defeq \{a, a+1, \ldots, b\}$ and for a positive integer $a$, we let $[a] \defeq [1..a]$.

Throughout the paper, each graph is finite, undirected and simple. Let $G$ be a graph. We denote the vertex set of $G$ by $V(G)$ and the edge set of $G$ by $E(G) \subseteq \binom{V}{2}$. For a vertex $v \in V(G)$, we denote by $N(v)$ the \emph{(open) neighborhood} of $G$, i.e., $N(v) \defeq \{w \mid \{v, w\} \in E(G)\}$. The \emph{degree of $v$} is the size of the neighborhood of $v$, i.e., $\deg(v) \defeq \card{N(v)}$. We define the \emph{closed neighborhood} of $v$ as $N[v] \defeq N(v) \cup \{v\}$. For a set of vertices $W \subseteq V(G)$, we let $N(W) \defeq \bigcup_{w \in W} N(w)$ and $N[W] \defeq N(W) \cup W$. For a set of vertices $\{v_1, \ldots, v_r\}$, we use the shorthand $N(v_1, \ldots, v_r) \defeq N(\{v_1, \ldots, v_r\})$.
	A vertex set $C \subseteq V(G)$ is called a \emph{clique}, if for each pair of distinct vertices $c_1, c_2 \in C$, $\{c_1, c_2\} \in E(G)$. A vertex set $I \subseteq V(G)$ is called \emph{independent}, if for each pair of distinct vertices $v_1, v_2 \in I$, $\{v_1, v_2\} \notin E(G)$. A graph $G$ is called \emph{bipartite}, if there is a partition $(X, Y)$ of its vertex set such that $X$ and $Y$ are independent.
	
	For two graphs $G$ and $H$, we denote by $H \subseteq G$ that $H$ is a \emph{subgraph} of $G$, i.e.\ that $V(H) \subseteq V(G)$ and $E(H) \subseteq E(G)$. For a vertex set $X \subseteq V(G)$, we denote by $G[X]$ the subgraph of $G$ \emph{induced} by $X$, i.e., $G[X] \defeq (X, E(G) \cap \binom{X}{2})$. We let $G - X \defeq G[V(G) \setminus X]$ and we use the shorthand $G - v$ for $G - \{v\}$. 
	For two disjoint vertex subsets $X, Y \subseteq V(G)$, we denote by $G[X, Y]$ the bipartite subgraph of $G$ induced by $(X, Y)$, that is $G[X, Y] \defeq (X \cup Y, \{\{x, y\} \in E(G) \mid x \in X, y \in Y\})$.
		
	
	A subgraph $P \subseteq G$ is called a \emph{path} if all its vertices have degree at most two in $P$ and there are precisely two distinct vertices in $V(P)$ that have degree one in $P$, called the \emph{endpoints} of $P$. For $s, t \in V(G)$, a path is called \emph{$(s, t)$-path} if it is a path with endpoints $s$ and $t$.
	
	We call two edges $e, f \in E(G)$ \emph{adjacent} if they share an endpoint, i.e., if there exist vertices $v, w, x \in V(G)$ such that $e = \{v, w\}$ and $f = \{v, x\}$. A \emph{matching} is a set of pairwise non-adjacent edges. 
	We say that a matching $M$ \emph{saturates} a set of vertices $W \subseteq V(G)$, if for all $v \in W$, there is a pair $\{v, w\} \in M$.
	
	Given a set of vertices $X \subseteq V(G)$, we call the operation of adding to $G$ a new vertex $x$ with neighborhood $N(X)$ and deleting all vertices in $X$ the \emph{contraction} of $X$.
	
\paragraph{Exponential-Time Hypothesis ($\ETH$).} In 2001, Impagliazzo and Paturi made a conjecture about the complexity of \textsc{$3$-Sat}, the problem of determining whether a given Boolean formula in conjunctive normal form with clauses of size at most $3$ has a satisfying assignment. This conjecture is known as the \emph{Exponential-Time Hypothesis ($\ETH$)} and has lead to a plethora of conditional lower bounds, see, e.g., the survey~\cite{FellowsEtAlTwo_LMS13} or~\cite[Chapter 14]{FellowsEtAlTwo_CFK15}. Formally, $\ETH$ can be stated as:\footnote{The $\cO^*$-notation suppresses polynomial factors in $n$.} 
\begin{conjecture}[$\ETH$~\cite{FellowsEtAlTwo_IP01,FellowsEtAlTwo_IPZ01}]
	There is an $\varepsilon > 0$ such that \textsc{$3$-Sat} on $n$ variables cannot be solved in time $\cO^*(2^{\varepsilon n})$.
\end{conjecture}

	\section{Standard Methods}\label{FellowsEtAlTwo_sec:standard:methods}
	In this section, we review some classic results in \textsc{Vertex Cover} kernelization. In particular, we discuss the Buss kernel~\cite{FellowsEtAlTwo_BG93} in \cref{FellowsEtAlTwo_sec:Buss}.  \Cref{FellowsEtAlTwo_sec:crown} is devoted to the kernel based on the notion of a crown decomposition \cite{FellowsEtAlTwo_CFJ04,FellowsEtAlTwo_Fel03} (see \cref{FellowsEtAlTwo_def:crown:dec}). A linear-programming-based kernel~\cite{FellowsEtAlTwo_CKJ01} is discussed in \cref{FellowsEtAlTwo_sec:lp:based}.
	
	We would like to remark that the technical parts of the expositions given in the remainder of this section are based on~\cite[Sections 2.2.1, 2.3 and 2.5]{FellowsEtAlTwo_CFK15} and we refer to this text for several details.
	
	\subsection{Buss Kernelization}\label{FellowsEtAlTwo_sec:Buss}
	The first kernel for \vc ~appeared several years before the notion of kernelization was formally introduced and is attributed to Buss~\cite{FellowsEtAlTwo_BG93}. It relies on two observations. The first one is that by definition, there is no need to include an isolated vertex in a vertex cover, as it does not have any incident edges that need to be covered.
	\begin{reduction}\label{FellowsEtAlTwo_reduction:isolated}
		If $G$ has an isolated vertex $v$, then reduce $(G, k)$ to $(G - v, k)$.
	\end{reduction}
	The second observation is that, if $G$ has a vertex $v$ of degree more than $k$, then we have no choice but to include $v$ in any size-$k$ vertex cover of $G$: If we did not include $v$, we would have to include all of its at least $k+1$ neighbors, exceeding the budget of $k$ vertices we are given. Hence, $G$ has a vertex cover of size $k$ if and only if $G - v$ has a vertex cover of size $k-1$, so we have observed that the following reduction rule is \emph{safe}, meaning that the original instance is a \textsc{Yes}-instance if and only if the reduced instance is a \textsc{Yes}-instance.
	\begin{reduction}\label{FellowsEtAlTwo_reduction:degree:k}
		If $G$ has a vertex $v$ with $\deg(v) > k$, then reduce $(G, k)$ to $(G - v, k - 1)$.
	\end{reduction}
	Now, after exhaustively applying \cref{FellowsEtAlTwo_reduction:degree:k}, $G$ has maximum degree at most $k$, so if $G$ contains more than $k^2$ edges, then we are dealing with a \textsc{No}-instance: It is not possible to cover more than $k^2$ edges with $k$ vertices of degree at most $k$. On the other hand, if $(G, k)$ is a \textsc{Yes}-instance, then $G$ has a vertex cover $X$ of size at most $k$. After exhaustively applying \cref{FellowsEtAlTwo_reduction:isolated}, $G$ does not contain any isolated vertices so we can assume that every vertex of $V(G) \setminus X$ has a neighbor in $X$. Since the maximum degree of $G$ is at most $k$, we can conclude that $\card{V(G) \setminus X} \le k^2$, which implies that $\card{V(G)} \le k^2 + k$. Hence, if $G$ has more than $k^2+k$ vertices, we can again conclude that we are dealing with a \textsc{No}-instance. Since \cref{FellowsEtAlTwo_reduction:isolated,FellowsEtAlTwo_reduction:degree:k} clearly run in polynomial time, we have the following theorem.
	\begin{theorem}[Buss and Goldsmith \cite{FellowsEtAlTwo_BG93}]
		\vc~admits a kernel with at most $k^2 + k$ vertices and $k^2$ edges.
	\end{theorem}
	
	\subsection{Crown Reduction}\label{FellowsEtAlTwo_sec:crown}
	The key insight above was that any vertex of degree at least $k+1$ has to be contained in any size-$k$ vertex cover of a graph. The kernel we present in this section follows a similar motivation. The goal is to identify a set of vertices that we can always \emph{assume} to be contained in a size-$k$ vertex cover of a graph. In other words, we want to find a set of vertices $S$, such that if $G$ contains a vertex cover of size $k$ then $G$ contains a vertex cover of size $k$ that contains $S$. The process of identifying such a set $S$ is based on a structural decomposition of the input graph, called the \emph{crown decomposition}.
	Formally, a crown decomposition is defined as follows and we illustrate it in \cref{FellowsEtAlTwo_fig:crown:dec}.
	\begin{definition}[Crown Decomposition]\label{FellowsEtAlTwo_def:crown:dec}
		Let $G$ be a graph. A \emph{crown decomposition} of $G$ is a partition $(C, H, B)$ of $V(G)$, where $C$ is called the \emph{crown}, $H$ the \emph{head} and $B$ the \emph{body}, such that the following hold.
		\begin{enumerate}[label=(\roman*)]
			\item $C$ is a non-empty independent set in $G$.\label{FellowsEtAlTwo_def:crown:dec:i:s}
			\item There are no edges between vertices in $C$ and vertices in $B$.\label{FellowsEtAlTwo_def:crown:dec:sep}
			\item $G[C, H]$ contains a matching that saturates $H$.\label{FellowsEtAlTwo_def:crown:dec:matching}
		\end{enumerate}
	\end{definition}
	\begin{figure}[t]
		\centering
		\includegraphics[height=.135\textheight]{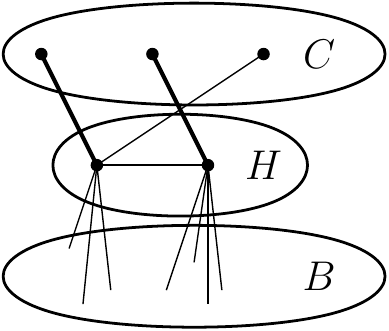}
		\caption{Illustration of a crown decomposition (\cref{FellowsEtAlTwo_def:crown:dec}). The bold edges in $G[C, H]$ show a matching saturating $H$.}
		\label{FellowsEtAlTwo_fig:crown:dec}
	\end{figure}
	The motivation for using the above definition in \vc~kernelization is as follows. Suppose we are given a crown decomposition $(C, H, B)$ of $G$ and consider the bipartite graph $G[C, H]$. Clearly, any vertex cover of $G$ has to cover the edges in $G[C \cup H]$. However, by~\cref{FellowsEtAlTwo_def:crown:dec:matching} we know that there is a matching in $G[C, H]$ saturating $H$, hence any vertex cover of $G[H, C]$ has size at least $\card{H}$. On the other hand, $H$ is a vertex cover of $G[C, H]$ and since $C$ is independent by~\cref{FellowsEtAlTwo_def:crown:dec:i:s}, of $G[C \cup H]$. This allows us to conclude that $G$ has a vertex cover of size $k$ if and only if $G - (C \cup H)$ has a vertex cover of size $k - \card{H}$. Hence, the following reduction rule is safe.
	\begin{reduction}\label{FellowsEtAlTwo_reduction:crown}
		If $G$ has a crown decomposition $(C, H, B)$, then reduce $(G, k)$ to $(G - (C \cup H), k - \card{H})$.
	\end{reduction}
	However, two questions remain. Namely whether we can find a crown decomposition of a graph in polynomial time and how to obtain the linear bound on the number of vertices in the resulting kernel. Both questions are answered by the following lemma whose proof is based on classic results in graph theory by K\"{o}nig~\cite{FellowsEtAlTwo_Koe16} and Hall~\cite{FellowsEtAlTwo_Hal35}, and polynomial-time algorithms for bipartite matching such as the classic algorithm due to Hopcroft and Karp~\cite{FellowsEtAlTwo_HK73}.\footnote{For a more fine-grained analysis one could apply the faster algorithm~\cite{FellowsEtAlTwo_MS04}.}
	\begin{lemma}[Lemma 2.14 in~\cite{FellowsEtAlTwo_CFK15} based on~\cite{FellowsEtAlTwo_CFJ04}]\label{FellowsEtAlTwo_lem:crown}
		Let $G$ be a graph on at least $3k+1$ vertices. There is a polynomial-time algorithm that either
		\begin{enumerate}[label=\arabic*)]
			\item finds a matching of size at least $k+1$ in $G$; or\label{FellowsEtAlTwo_lem:crown:matching}
			\item finds a crown decomposition of $G$.\label{FellowsEtAlTwo_lem:crown:decomposition}
		\end{enumerate}
	\end{lemma}	
	Now, in Case~\cref{FellowsEtAlTwo_lem:crown:matching} we can immediately conclude that $(G, k)$ is a \textsc{No}-instance and in Case~\cref{FellowsEtAlTwo_lem:crown:decomposition} we can apply \cref{FellowsEtAlTwo_reduction:crown}. By an exhaustive application of \cref{FellowsEtAlTwo_lem:crown} in combination with \cref{FellowsEtAlTwo_reduction:crown} (and \cref{FellowsEtAlTwo_reduction:isolated} to get rid of isolated vertices), we have the following theorem.
	\begin{theorem}[Chor et al.\ \cite{FellowsEtAlTwo_CFJ04}]
		\vc~admits a kernel with at most $3k$ vertices.
	\end{theorem}
	We would like to remark that recently, a kernel on $2k$ vertices that only uses crown decomposition was obtained~\cite{FellowsEtAlTwo_LS18}.
	
	\subsection{LP-Based Kernel}\label{FellowsEtAlTwo_sec:lp:based}
	The \vc ~problem is one of many $\NP$-hard problems that can be expressed as an integer linear program~\cite{FellowsEtAlTwo_Sch98}, a fact which is commonly exploited in the field of approximation algorithms~\cite{FellowsEtAlTwo_Vaz03}. In this section, we show how to use linear programming to obtain a kernel for \textsc{Vertex Cover} on at most $2k$ vertices. We first recall how to formulate \vc ~as an integer linear program.
	
	For each vertex $v \in V(G)$, we introduce a variable $x_v \in \{0, 1\}$ with the interpretation that $x_v = 1$ if and only if the vertex $v$ is included in the vertex cover witnessed by a solution to the (integer) linear program. We can then formulate the constraints in a natural way, directly applying the definition of vertex covers: For each edge $uv \in E(G)$, the requirement that at least one of $u$ and $v$ has to be contained in the solution translates to the constraint $x_u + x_v \ge 1$. Since we are looking for a vertex cover of minimum size, the objective function minimizes the sum over all $x_v$'s. 
	\begin{align}
		&\min \sum_{v \in V(G)}x_v \nonumber \\
		\mbox{subject to~~~} &x_u + x_v \ge 1 ~~~\forall uv \in E(G) \label{FellowsEtAlTwo_eq:lp:edges} \\
		&x_v \in \{0, 1\} ~~~~\forall v \in V(G) \label{FellowsEtAlTwo_eq:lp:integer}
	\end{align}
	To make the program feasible to compute, we relax the integrality constraints \cref{FellowsEtAlTwo_eq:lp:integer} to $x_v \in \bR$, $x_v \ge 0$. (Note that we can drop the constraints $x_v \le 1$ since the objective function is a minimization.) The resulting linear program is solvable in polynomial time, but may not always return a feasible solution for the original \textsc{Vertex Cover} instance. However, we are chasing a different goal here, a kernelization algorithm.
	
	Given an optimal solution $(x_v)_{v \in V(G)}$ of the (relaxed) linear program, we define the sets 
		$V_0 \defeq \{v \in V(G) \mid x_v < \frac{1}{2}\}$,
		$V_1 \defeq \{v \in V(G) \mid x_v > \frac{1}{2}\}$, and 
		$V_{\frac{1}{2}} \defeq \{v \in V(G) \mid x_v = \frac{1}{2}\}$.
	The key ingredient is the following theorem due to Nemhauser and Trotter~\cite{FellowsEtAlTwo_NT74}.
	\begin{theorem}[Nemhauser and Trotter \cite{FellowsEtAlTwo_NT74}]\label{FellowsEtAlTwo_thm:nem:trott}
		There is a minimum vertex cover $X$ of $G$ such that $V_1 \subseteq X \subseteq V_1 \cup V_{\frac{1}{2}}$.
	\end{theorem}
	We derive a reduction rule from \cref{FellowsEtAlTwo_thm:nem:trott}. First, we note that in any \textsc{Yes}-instance of \vc, $\sum_{v \in V(G)} x_v \le k$. Furthermore, let $X$ be a vertex cover of $G$ of size $k$ with $V_1 \subseteq X$ and $X \cap V_0 = \emptyset$ (whose existence is guaranteed by \cref{FellowsEtAlTwo_thm:nem:trott}), then $X \setminus V_1$ is a vertex cover of $G - (V_0 \cup V_1)$ of size $k - \card{V_1}$. Conversely, if $G - (V_0 \cup V_1)$ has a vertex cover $X'$ of size $k'$, we observe that by the constraints \cref{FellowsEtAlTwo_eq:lp:edges}, for any edge $vw \in E(G)$ with $v \in V_0$, we have that $w \in V_1$. Hence, $X' \cup V_1$ is a vertex cover of $G$ of size $k' + \card{V_1}$. We have argued that the following reduction rule is safe.
	\begin{reduction}\label{FellowsEtAlTwo_reduction:lp}
		Let $(x_v)_{v \in V(G)}$, $V_0$, $V_{\frac{1}{2}}$ and $V_1$ be as above. If $\sum_{v \in V(G)} x_v > k$, then conclude that we are dealing with a \textsc{No}-instance. Otherwise, reduce $(G, k)$ to $(G - V_0 \cup V_1, k - \card{V_1})$.
	\end{reduction}
	The number of vertices in the reduced instance after applying \cref{FellowsEtAlTwo_reduction:lp} is
	\begin{align*}
		\card{V(G) \setminus (V_0 \cup V_1)} = \card{V_\frac{1}{2}} = \sum\nolimits_{v \in V_{\frac{1}{2}}} 2x_v  \le 2 \cdot \sum\nolimits_{v \in V(G)} x_v \le 2k,
	\end{align*}	
	so we have obtained the following kernel for \vc.\footnote{We would like to remark that while \textsc{Linear Programming} can be solved in polynomial time (and hence our reduction runs in polynomial time), the corresponding algorithms are often slow in practice. However, for the case of \vc ~there is good news: One can show that a solution of the above linear program can be found via a reduction to \textsc{Bipartite Matching} (see, e.g., \cite[Section 2.5]{FellowsEtAlTwo_CFK15}) which has fast practical algorithms.}
	\begin{theorem}[Chen et al.\ \cite{FellowsEtAlTwo_CKJ01}]
		\vc~admits a kernel with at most $2k$ vertices.
	\end{theorem}
	
	\section{Towards the Barrier -- What is the Maximum Minimum Vertex Degree of the Kernel that Can be Achieved in Polynomial Time?}\label{FellowsEtAlTwo_sec:small:degree}	
	In the previous section, we have seen that by \cref{FellowsEtAlTwo_reduction:degree:k}  we can kernelize all vertices whose degree is larger than the target value $k$ of the given vertex cover instance. Hence,  after applying this rule exhaustively there will be no vertex of degree larger than $k$ in the kernelized instance. But what about vertices of small degree? Vertices of degree zero, i.e., isolated vertices, can be removed from a \vc ~instance according to \cref{FellowsEtAlTwo_reduction:isolated}. Furthermore, we will see below that there are fairly simple reduction rules that kernelize vertices of degree one and two (see \cref{FellowsEtAlTwo_reduction:pendant:edge,FellowsEtAlTwo_reduction:degree:two}). 
	A natural question arises: Can this process go on `forever', i.e., can we, for any fixed constant $d \in \bN$, give a reduction rule that kernelizes all vertices of degree $d$ from a given \vc ~instance?
	
	The answer to this question is \emph{probably not} --- even if the degree of the polynomial in the runtime of the kernelization algorithm can depend on $d$: It is well-known (see, e.g., \cite{FellowsEtAlTwo_CFK15,FellowsEtAlTwo_DF13,FellowsEtAlTwo_Fel14}) that unless $\ETH$ fails, there is some barrier constant $\barrierconstant > 0$ such that the fastest possible algorithm for \textsc{Vertex Cover} runs in time $(1 + \barrierconstant)^k \cdot n^{\cO(1)}$. If we could kernelize \vc ~in polynomial time to arbitrarily large minimum degree, one could devise a straightforward branching algorithm that runs in time $(1 + \barrierconstant - \varepsilon)^k \cdot n^{\cO(1)}$, for some $0 < \varepsilon < \barrierconstant$, where $\varepsilon$ can be arbitrarily close to the value of $\barrierconstant$. We coin the corresponding integer $\barrierdegree \in \bN$ the \emph{barrier degree} of \vc ~kernelization and now prove formally its existence (assuming $\ETH$).
	
	\begin{proposition}\label{FellowsEtAlTwo_prop:barrier:degree}
		Unless $\ETH$ fails, there is some constant $\barrierdegree \in \bN$, such that \textsc{Vertex Cover} cannot be kernelized to instances of minimum degree $\barrierdegree$.
	\end{proposition}
	\begin{proof}
		Using standard arguments about branching algorithms (see, e.g., \cite[Chapter 3]{FellowsEtAlTwo_CFK15}) one can show that there is an algorithm solving vertex cover in time $\lambda^k \cdot n^{\cO(1)}$, where $\lambda$ satisfies
		\begin{align}
			\lambda \le \lambda^d(\lambda-1), \label{FellowsEtAlTwo_eq:branching:runtime}
		\end{align}
		if the input graph always has a vertex of degree at least $d$ to branch on. Now suppose that the statement of the proposition is false, then we can guarantee the existence of such a vertex for constant but arbitrarily large $d$ (with only polynomial time overhead at each stage of the branching). Now let $\varepsilon > 0$ with $\varepsilon < \barrierconstant$. (Note that this implies that $\varepsilon < 1$ as $\barrierconstant < 0.2738$~\cite{FellowsEtAlTwo_CKX10}.) We substitute $\lambda$ with $(1 + \varepsilon)$ in \cref{FellowsEtAlTwo_eq:branching:runtime} and obtain:
		\begin{align*}
			1 + \varepsilon &\le (1 + \varepsilon)^d \varepsilon \iff (1+\varepsilon)^d \ge \frac{1+\varepsilon}{\varepsilon} \iff d \log(1+ \varepsilon) \ge \log\left(\frac{1+ \varepsilon}{\varepsilon}\right) \\
			\iff d &\ge \frac{\log\left(\frac{1+ \varepsilon}{\varepsilon}\right)}{\log(1+\varepsilon)} = \frac{\log(1+\varepsilon) - \log(\varepsilon)}{\log(1+\varepsilon)} = 1 - \frac{\log(\varepsilon)}{\log(1+\varepsilon)}.
		\end{align*}
		This shows that for any such $\varepsilon$, there is a constant $d_\varepsilon \in \bN$ such that, if we could kernelize \textsc{Vertex Cover} to minimum degree $d_\varepsilon$, then we could solve it in $(1 + \varepsilon)^k \cdot n^{\cO(1)}$ time, where $\varepsilon < \barrierconstant$ by our choice. This contradicts $\ETH$ by, e.g., \cite[Theorem 1]{FellowsEtAlTwo_Fel14}.
    \qed
	\end{proof}
	
	\begin{table}
		\centering
		\begin{tabular}{c||c|c|c|c|c|c}
			Chen et al.~\cite{FellowsEtAlTwo_CKX10} & $d = 5$ & $d = 6$ & $d = 7$ & $d = 10$ & $d = 25$ & $d = 100$ \\
			\hline 
			$1.2738^k$ & $1.3247^k$ & $1.2852^k$ & $1.2555^k$ & $1.1975^k$ & $1.1005^k$ & $1.0346^k$ 
		\end{tabular}
		\caption{(Dependence on $k$ of the) runtime of the resulting simple branching $\FPT$-algorithm when using a kernelization algorithm to minimum degree $d$, for several values of $d$, versus the current fastest known $\FPT$-algorithm for \vc~\cite{FellowsEtAlTwo_CKX10}.}
		\label{FellowsEtAlTwo_tab:kernel:degree:fpt}
	\end{table}	
	The proof of \cref{FellowsEtAlTwo_prop:barrier:degree} also provides some very natural motivation for the question of kernelizing \textsc{Vertex Cover} to larger and larger minimum degree; such kernels immediately provide new $\FPT$-algorithms for the problem. In particular, kernelizing to minimum degree seven would already improve upon the current best known algorithm for \textsc{Vertex Cover}, yielding first progress in a very attractive research question in over a decade! We illustrate the runtime of such algorithms for several concrete values of $d$ in \cref{FellowsEtAlTwo_tab:kernel:degree:fpt}.
		
	In the remainder of this section, we present a set of reduction rules that were first observed by Fellows and Stege~\cite{FellowsEtAlTwo_FS99} to kernelize a vertex cover instance to minimum degree `almost five', in the following sense: We show that a vertex can be kernelized if its degree is at most three or its degree is four and there are more than two edges between the vertices in its neighborhood.
	
	Before we give the reduction rules to kernelize vertices of degree one and two, we would like to remark that later in the text, we introduce	two \emph{auxiliary} reduction rules, mostly to deal with structures arising in the kernelization of vertices of degree three and four 
	which as a byproduct also kernelize degree one and two vertices. 
	For explanatory purposes, however, we describe the reduction rules for vertices of degree one and two separately first.
	
	\begin{reduction}\label{FellowsEtAlTwo_reduction:pendant:edge}
		If $G$ has a pendant edge $\{u, v\}$ with $\deg(u) = 1$, then reduce $(G, k)$ to $(G - \{u, v\}, k-1)$.
	\end{reduction}	
	\begin{proposition}
		\Cref{FellowsEtAlTwo_reduction:pendant:edge} is safe, i.e., if $G$ has a pendant edge $\{u, v\}$ with $\deg(u) = 1$, then $G$ has a vertex cover of size $k$ if and only if $G - \{u, v\}$ has a vertex cover of size $k-1$.
	\end{proposition}
	\begin{proof}
		($\Rightarrow$) Suppose $G$ has a vertex cover $X^*$ of size $k$. Since $\{u, v\}$ is an edge of $G$, at least one of $u$ and $v$ is contained in $X^*$. If $v \notin X^*$, then we let $X \defeq X^* \setminus \{u\} \cup \{v\}$. Note that $X$ is a vertex cover since $v$ is the only neighbor of $u$. If $v \in X^*$, we simply let $X \defeq X^*$. Since $v \in X$, $X \setminus \{v\}$ is a vertex cover of $G - \{u, v\}$ of size $k-1$. 
		
		($\Leftarrow$) Let $X'$ be a vertex cover of $G - \{u, v\}$ of size $k-1$. We observe that any edge in $E(G) \setminus E(G - \{u, v\})$ is incident with $v$ and conclude that $X' \cup \{v\}$ is a vertex cover of $G$ of size $k$.
		\qed
	\end{proof}
	\begin{figure}
		\begin{subfigure}{.38\textwidth}
			\centering
			\includegraphics[height=.125\textheight]{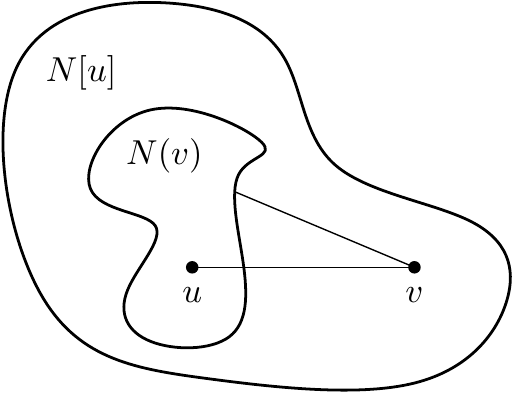}
			\caption{The situation of \cref{FellowsEtAlTwo_reduction:adjacent:vertices}.}
			\label{FellowsEtAlTwo_fig:reduction:adjacent:vertices}
		\end{subfigure}
		\hfill
		\begin{subfigure}{.6\textwidth}
			\centering
			\includegraphics[height=.105\textheight]{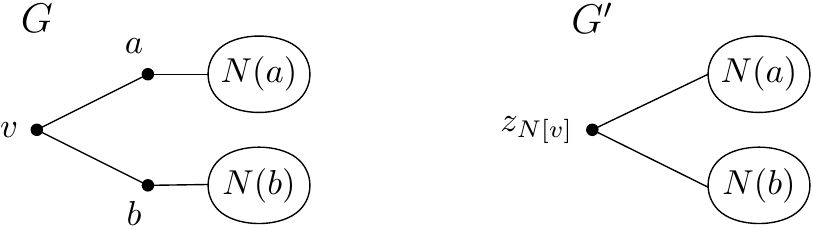}
			\caption{Illustration of \cref{FellowsEtAlTwo_reduction:degree:two}. Note that by \cref{FellowsEtAlTwo_reduction:adjacent:vertices}, we can assume that $a$ and $b$ are not adjacent.}
			\label{FellowsEtAlTwo_fig:reduction:degree:two}
		\end{subfigure}
		\caption{Illustrations of \cref{FellowsEtAlTwo_reduction:adjacent:vertices,FellowsEtAlTwo_reduction:degree:two}, respectively.}
	\end{figure}
	Before we show how to kernelize degree two vertices, we give the first auxiliary reduction rule.
	\begin{reduction}\label{FellowsEtAlTwo_reduction:adjacent:vertices}
		If $G$ has two adjacent vertices $u$ and $v$ such that $N(v) \subseteq N[u]$, then reduce $(G, k)$ to $(G - u, k - 1)$.
	\end{reduction}
	For an illustration of the situation of \cref{FellowsEtAlTwo_reduction:adjacent:vertices}, see \cref{FellowsEtAlTwo_fig:reduction:adjacent:vertices}.
	\begin{proposition}
		\Cref{FellowsEtAlTwo_reduction:adjacent:vertices} is safe, i.e., if $G$ has two adjacent vertices $u$ and $v$, and $N(v) \subseteq N[u]$, then $G$ contains a vertex cover of size $k$ if and only if $G - u$ has a vertex cover of size $k - 1$.
	\end{proposition}
	\begin{proof}
		($\Rightarrow$) Suppose $G$ has a vertex cover $X^*$ of size $k$. If $u \notin X^*$ then $N(u)$ must be in $X^*$, so by assumption, it contains
		 $N[v] \setminus \{u\}$. But then, $X \defeq X^* \setminus \{v\} \cup \{u\}$ is also a vertex cover of $G$ of size $k$, so we can assume that $u \in X$. Then, $X \setminus \{u\}$ is a vertex cover of $G - u$ of size $k-1$.
 	
 	($\Leftarrow$) is immediate since for any vertex cover $X'$ of $G - u$, $X' \cup \{u\}$ is a vertex cover of $G$.
	\qed
	\end{proof}
	
	The next reduction rule takes care of vertices of degree two and is illustrated in \cref{FellowsEtAlTwo_fig:reduction:degree:two}.	
	\begin{reduction}\label{FellowsEtAlTwo_reduction:degree:two}
		If \cref{FellowsEtAlTwo_reduction:adjacent:vertices} cannot be applied and $G$ has a vertex $v$ with $\deg(v) = 2$, then reduce $(G, k)$ to $(G', k-1)$, where $G'$ is the graph obtained from $G$ by contracting $N[v]$ to a single vertex.
	\end{reduction}
	\begin{proposition}
		\Cref{FellowsEtAlTwo_reduction:degree:two} is safe, i.e., under its stated conditions, $G$ has a vertex cover of size $k$ if and only if $G'$ has a vertex cover of size $k-1$.
	\end{proposition}
	\begin{proof}
		Throughout the proof, we denote the neighborhood of $v$ in $G$ by $N(v) = \{a, b\}$ and the vertex in $G'$ that was created due to the contraction of $N[v]$ by $z_{N[v]}$. We can assume that $\{a, b\} \notin E(G)$: If the edge $\{a, b\}$ was present, then $N(v) \subseteq N[a]$ (and $N(v) \subseteq N[b]$), so we could have applied \cref{FellowsEtAlTwo_reduction:adjacent:vertices}.
		
		($\Rightarrow$) We observe that each edge in $E(G') \setminus E(G)$ has an endpoint in $\{z_{N[v]}\} \cup N(a, b)$. Let $X$ be a vertex cover of $G$ of size $k$. If $X \cap N(v) = \emptyset$, then $\{v\} \cup N(a, b) \subseteq X$ and we can conclude that $X \setminus \{v\}$ is a vertex cover of $G'$. If $N(v) \subseteq X$, then $X' \defeq X \setminus N(v) \cup \{z_{N[v]}\}$ is a vertex cover of $G'$. (Note that in this case, $X'$ has size at most $k-1$ as well.) If $X$ contains precisely one vertex from $N(v)$, assume w.l.o.g.\ that $X \cap N(v) = \{a\}$, then $v \in X$ (otherwise the edge $\{v, b\}$ is not covered), so $X \setminus \{v, a\} \cup \{z_{N[v]}\}$ is a vertex cover of $G'$ of size $k-1$.
		
		($\Leftarrow$) Let $X'$ be a vertex cover of $G'$ of size $k-1$. We distinguish the cases when $z_{N[v]} \in X'$ and when $z_{N[v]} \notin X'$. In the former case, $X' \setminus \{z_{N[v]}\} \cup \{a, b\}$ is a vertex cover of $G$, since each edge in $E(G) \setminus E(G')$ is incident with a vertex in $\{a, b\}$. In the latter case, $N(a, b) \subseteq X'$ since $z_{N[v]} \notin X'$, and we have that $X' \cup \{v\}$ is a vertex cover of $G$: Since $\{a, b\} \notin E(G)$, each edge in $E(G) \setminus E(G')$ is incident with a vertex in $\{v\} \cup N(a, b)$. In both cases, the size of the resulting vertex cover is $k$.
		\qed
	\end{proof}
	
	\newcommand\degreev{\alpha}
	\newcommand\degreevhalf{\frac{\degreev}{2}}	
	
	Before we proceed with kernelizing vertices of degree larger than two, we require one more auxiliary reduction rule. 
	This reduction rule will be crucially used to argue that we can exclude certain structures appearing in the subgraphs induced by the neighborhoods of small-degree vertices. 
	It captures~\cite[Reductions R.4 and R.5]{FellowsEtAlTwo_FS99} and is illustrated in \cref{FellowsEtAlTwo_fig:reduction:almost:clique:neighborhood}. Note that due to its complexity, it will only be executed for vertices whose degree is bounded by a fixed constant $\degreev$ (independent of $k$). In particular, for our purposes it will be sufficient to make use of the following reduction for $\degreev \le 4$.
	
	\begin{reduction}\label{FellowsEtAlTwo_reduction:cliques:co:matching:new}
		Suppose $G$ has a vertex $v$ such that the following hold. There is a partition $(C_1, C_2)$ of $N_G(v)$ where $\card{C_1} \ge \card{C_2}$ and the following hold.
		\begin{enumerate}[label=(\roman*)]
			\item $C_i$ is a clique for all $i \in [2]$.\label{FellowsEtAlTwo_reduction:cliques:co:matching:new:cliques}
			\item Let $M$ be the set of non-edges of $G[C_1, C_2]$. For each $c_1 \in C_1$, there is \emph{precisely one} $f \in M$ such that $c_1 \in f$.\label{FellowsEtAlTwo_reduction:cliques:co:matching:new:m}
		\end{enumerate}
		Then, reduce $(G, k)$ to $(G', k - \card{C_2})$, where $G'$ is obtained from $G$ by 
		\begin{enumerate}[label=(\arabic*.)]
			\item deleting $v$ and $C_2$, and
			\item for all $\{c_1, c_2\} \in M$ with $c_1 \in C_1$ and $c_2 \in C_2$,  adding all edges between $c_1$ and $N_G(c_2)$.
		\end{enumerate}
	\end{reduction}
	\begin{figure}[t]
		\centering
		\includegraphics[width=.8\textwidth]{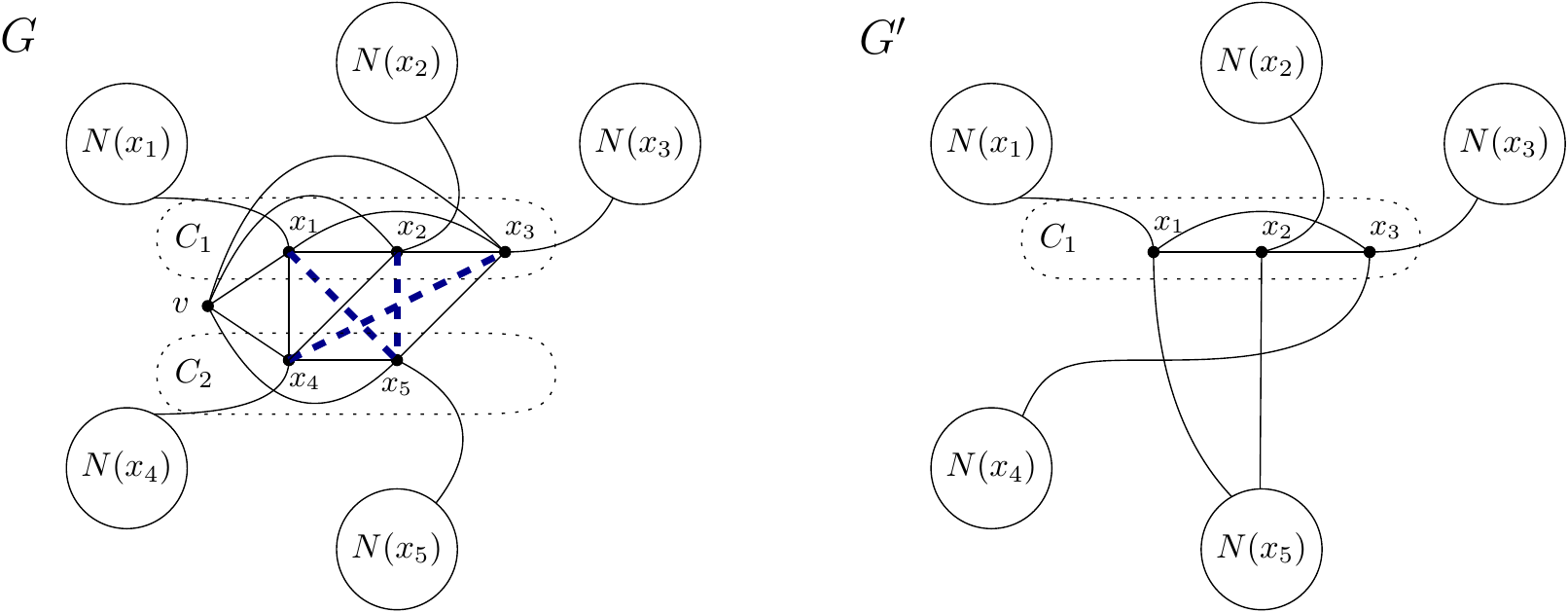}
		\caption{Illustration of \cref{FellowsEtAlTwo_reduction:cliques:co:matching:new}. Note that $\card{C_1} > \card{C_2}$, and the bold dotted lines between vertices of $C_1$ and $C_2$ are the set of non-edges $M$ in $G[C_1, C_2]$ satisfying condition \ref{FellowsEtAlTwo_reduction:cliques:co:matching:new:m}: For every vertex $x \in C_1$ there is precisely one element in $M$ containing $x$.}
		\label{FellowsEtAlTwo_fig:reduction:almost:clique:neighborhood}
	\end{figure}	
	\begin{proposition}\label{FellowsEtAlTwo_prop:cliques:co:matching:new}
		\Cref{FellowsEtAlTwo_reduction:cliques:co:matching:new} is safe, i.e.\ under its stated conditions, $G$ has a vertex cover of size $k$ if and only if $G'$ has a vertex cover of size $k - \card{C_2}$.
	\end{proposition}
	\begin{proof}
		Since by assumption \cref{FellowsEtAlTwo_reduction:cliques:co:matching:new:cliques} of \cref{FellowsEtAlTwo_reduction:cliques:co:matching:new}, $C_1$ and $C_2$ are cliques in $G$, and since $C_1$ remains a clique in $G'$, we make the following observation.
		\begin{observation}\label{FellowsEtAlTwo_obs:cliques:co:matching:new:all:cases}
			Every vertex cover of $G$ contains at least $\card{C_i} - 1$ vertices from $C_i$ for all $i \in [2]$, and every vertex cover of $G'$ contains at least $\card{C_1} - 1$ vertices from $C_1$.
		\end{observation}
		We now prove the proposition by a case analysis on the structure of the intersection of vertex covers of $G$ and $G'$ with $N_G(v) = C_1 \cup C_2$ and $C_1$, respectively. \Cref{FellowsEtAlTwo_obs:cliques:co:matching:new:all:cases} will be used later to argue that we covered all possible cases.
		\begin{claim}\label{FellowsEtAlTwo_claim:clique:co:matching:new:1}
			$G$ contains a vertex cover $X$ of size $k$ such that $N_G(v) \subseteq X$ if and only if $G'$ contains a vertex cover $X'$ of size $k - \card{C_2}$ such that $C_1 \subseteq X'$. 
		\end{claim}
		\begin{clproof}
			($\Rightarrow$) Let $X$ be a vertex cover of $G$ of size $k$ such that $N_G(v) \subseteq X$. (Note that we can assume that $v \notin X$.) We have that $X' \defeq X \setminus C_2$ is a vertex cover of $G^* \defeq G - (\{v\} \cup C_2)$. By construction, any edge in $E(G') \setminus E(G^*)$ is incident with a vertex in $C_1 \subseteq X'$, so $X'$ is a vertex cover of $G'$. Clearly, $\card{X'} = k - \card{C_2}$.
			
			($\Leftarrow$) Let $X'$ be a vertex cover of $G'$ of size $k - \card{C_2}$ such that $C_1 \subseteq X'$. Then, $X \defeq X' \cup C_2$ is a vertex cover of $G$, since every edge in $E(G) \setminus E(G')$ is either incident with a vertex in $C_2$ or with $v$. For the latter case, we observe that $N_G(v) = (C_1 \cup C_2) \subseteq X$. Clearly, $\card{X} = k$.
      \claimqed
		\end{clproof}
		We observe that \cref{FellowsEtAlTwo_claim:clique:co:matching:new:1} also covers the case when a size-$k$ vertex cover of $G$ misses precisely one vertex from $N_G(v)$: Let $X^*$ be such a vertex cover and let $c \in N_G(v) \setminus X$. Since $X^*$ has to contain an endpoint of the edge $\{v, c\}$ and $c \notin X^*$, we can conclude that $v \in X^*$. Now, we simply let $X \defeq X^* \setminus \{v\} \cup \{c\}$ and observe that $X$ is a vertex cover of $G$ of size $k$ such that $N_G(v) \subseteq X$.
		
		\begin{claim}\label{FellowsEtAlTwo_claim:clique:co:matching:new:2}
			$G$ contains a vertex cover $X$ of size $k$ with $\card{C_i \setminus X} = 1$ for all $i \in [2]$ if and only if $G'$ contains a vertex cover $X'$ of size $k - \card{C_2}$ with $\card{C_1 \setminus X'} = 1$.
		\end{claim}
		\begin{clproof}
			($\Rightarrow$) Let $X$ be a size-$k$ vertex cover of $G$ such that for all $i \in [2]$, $\card{C_i \setminus X} = 1$ and let $c_i \in C_i \setminus X$ be the unique vertex in $C_i$ that is not contained in $X$. First, since $c_i \notin X$, we have that $v \in X$, otherwise the edge $\{v, c_i\}$ is not covered by $X$. Furthermore, we can conclude that $\{c_1, c_2\} \notin E(G)$, since if $\{c_1, c_2\}$ was an edge of $G$, then this edge was not covered by $X$. Clearly, since $\{c_1, c_2\} \notin E(G)$, we have that $\{c_1, c_2\} \notin E(G[C_1, C_2])$, so $\{c_1, c_2\} \in M$.
			
			We have argued that $\{c_1, c_2\} \in M$, and by condition \cref{FellowsEtAlTwo_reduction:cliques:co:matching:new:m} of \cref{FellowsEtAlTwo_reduction:cliques:co:matching:new} we know that $\{c_1, c_2\}$ is the only element in $M$ that contains $c_1$.
			We now show that 
			$$X' \defeq X \cap V(G') = X \setminus (\{v\} \cup C_2)$$ 
			is a vertex cover of $G'$. Clearly, $X'$ is a vertex cover of $G^* \defeq G - (\{v\} \cup C_2)$. Now, consider an edge $e' \in E(G') \setminus E(G^*)$. By construction, one of the endpoints of $e'$, say $x$, is from $C_1$. If $x \neq c_1$, then the edge $e'$ is covered by $X'$, since $C_1 \setminus \{c_1\} \subseteq X'$. 
		Now suppose that $x = c_1$ and denote the other endpoint of $e'$ by $y$. Since $e' \in E(G') \setminus E(G^*)$, following the construction of \cref{FellowsEtAlTwo_reduction:cliques:co:matching:new}, we can conclude that there is some $\{c_1, z\} \in M$ such that $y \in N_G(z)$. We can infer that $z = c_2$, since \cref{FellowsEtAlTwo_reduction:cliques:co:matching:new:m} asserts that there is only one element in $M$ that contains $c_1$ and we know by the above argument that $\{c_1, c_2\} \in M$. As $X$ is a vertex cover of $G$ and $c_2 \notin X$ by assumption, we know that $y \in X$, and so:
		$$y \in N_G(c_2) \cap V(G') \subseteq X \cap V(G') = X',$$
		hence the edge $\{c_1, y\}$ is covered by $X'$. We can conclude that $X'$ is a vertex cover of $G'$.
			Since we obtained $X'$ from $X$ by removing from it the vertex $v$ and $\card{C_2} - 1$ vertices from $C_2$, we have that $\card{X'} = k - \card{C_2}$. Clearly, $\card{C_1 \setminus X'} = 1$.
			
			($\Leftarrow$) Let $X'$ be a vertex cover of $G'$ of size $k - \card{C_2}$ such that $\card{C_1 \setminus X'} = 1$ and denote by $c_1 \in C_1 \setminus X'$ the unique vertex of $C_1$ that is not contained in $X'$. Let furthermore $c_2 \in C_2$ be such that $\{c_1, c_2\} \in M$. By condition \cref{FellowsEtAlTwo_reduction:cliques:co:matching:new:m}, such a vertex $c_2$ exists and it is unique.
			We argue that $X \defeq X' \cup \{v\} \cup (C_2 \setminus \{c_2\})$ is a vertex cover of $G$. Suppose for a contradiction that there is an edge $e \in E(G)$ that is not covered by $X$. Since $X'$ is a vertex cover of $G'$ and $X \supseteq X'$, we have that $e \in E(G) \setminus E(G')$. By construction, each such edge $e$ has (at least) one endpoint in $\{v\} \cup C_2$. Since $\{v\} \cup (C_2 \setminus \{c_2\}) \subseteq (N_G[v] \setminus \{c_1, c_2\}) \subseteq X$, we can conclude that $c_2$ is an endpoint of $e$, and the endpoint of $e$ other than $c_2$, say $y$, is contained in $N_G(c_2) \setminus (N_G[v] \setminus \{c_1\})$. Since $\{c_1, c_2\} \in M$, it is a non-edge, and so $y \in N_G(c_2) \setminus N_G[v]$.
		From the construction of \cref{FellowsEtAlTwo_reduction:cliques:co:matching:new}, we infer that $y \in N_G(c_2) \cap V(G') \subseteq N_{G'}(c_1) \cap V(G')$.
		
		To conclude, since $X \cap V(G') = X' \cap V(G')$, we have that if $y \notin X$, then $y \notin X'$. Since $y \in N_{G'}(c_1)$ and $c_1 \notin X'$ by assumption, the edge $\{c_1, y\}$ is not covered by $X'$, a contradiction with the assumption that $X'$ was a vertex cover of $G'$.
		It is clear that $\card{X} = k$ and that for all $i \in [2]$, $\card{C_i \setminus X} = 1$. \claimqed
		\end{clproof}
		We are now ready to finalize the proof of safeness of \cref{FellowsEtAlTwo_reduction:cliques:co:matching:new}. Suppose $G$ has a vertex cover $X$ of size $k$. Then, by \cref{FellowsEtAlTwo_obs:cliques:co:matching:new:all:cases}, we are in one of the following cases: (I) $N_G(v) \subseteq X$, (II) $\card{N_G(v) \setminus X} = 1$, or (III) for all $i \in [2]$, $\card{C_i \setminus X} = 1$. In cases (I) and (II), we can conclude that $G'$ has a vertex cover of size $k - \card{C_2}$ by \cref{FellowsEtAlTwo_claim:clique:co:matching:new:1} (and the remark thereafter). In case (III), $G'$ has a vertex cover of size $k - \card{C_2}$ by \cref{FellowsEtAlTwo_claim:clique:co:matching:new:2}.
		
		For the other direction, suppose $G'$ has a vertex cover $X'$ of size $k - \card{C_2}$. Again by \cref{FellowsEtAlTwo_obs:cliques:co:matching:new:all:cases}, we are in one of the following two cases: (IV) $C_1 \subseteq X'$, or (V) $\card{C_1 \setminus X'} = 1$.
		In case (IV), we can use \cref{FellowsEtAlTwo_claim:clique:co:matching:new:1} to conclude that $G$ has a vertex cover of size $k$ and in case (V) we can use \cref{FellowsEtAlTwo_claim:clique:co:matching:new:2}. This finishes the proof of \cref{FellowsEtAlTwo_prop:cliques:co:matching:new}. \qed
	\end{proof}
	
	Before we turn to kernelizing degree-three vertices, we observe that a combination of \cref{FellowsEtAlTwo_reduction:isolated,FellowsEtAlTwo_reduction:adjacent:vertices,FellowsEtAlTwo_reduction:cliques:co:matching:new} kernelizes vertices of degree one and two as well. Suppose $v$ is a vertex of degree one in $G$ whose only neighbor is $u$. Then, $N(v) = \{u\} \subseteq N[u]$, so following \cref{FellowsEtAlTwo_reduction:adjacent:vertices}, we could have removed the vertex $u$ and decreased the parameter value by one. In $G - u$, the vertex $v$ is an isolated vertex, so by \cref{FellowsEtAlTwo_reduction:isolated} it can be removed. These two steps together have the same effect as an application of \cref{FellowsEtAlTwo_reduction:pendant:edge}, the rule for kernelizing vertices of degree one. 
	
	Next, suppose that $v$ is a vertex of degree two and let $\{a, b\} \defeq N(v)$. There are two cases we have to consider.
	If $\{a, b\} \in E(G)$, then $N(v) \subseteq N[a]$ and we could have applied \cref{FellowsEtAlTwo_reduction:adjacent:vertices}. In the resulting instance whose graph is $G - a$, the vertex $v$ is of degree one so it would be removed by a combination of \cref{FellowsEtAlTwo_reduction:isolated,FellowsEtAlTwo_reduction:adjacent:vertices}, following the same argument as above.
	If $\{a, b\} \notin E(G)$, then $v$ trivially satisfies the conditions of \cref{FellowsEtAlTwo_reduction:cliques:co:matching:new} by considering the partition of $\{a, b\}$ into parts $\{a\}$ and $\{b\}$. Applying \cref{FellowsEtAlTwo_reduction:cliques:co:matching:new}, we can remove the vertex $v$.
	\begin{observation}
		After an exhaustive application of
		\cref{FellowsEtAlTwo_reduction:isolated,FellowsEtAlTwo_reduction:adjacent:vertices,FellowsEtAlTwo_reduction:cliques:co:matching:new} (where for \cref{FellowsEtAlTwo_reduction:cliques:co:matching:new}, $\degreev = 2$), the resulting graph has minimum degree three.
	\end{observation}	
	
	We are now ready to kernelize degree-three vertices.
	
	\begin{reduction}\label{FellowsEtAlTwo_reduction:degree:three}
		If neither \cref{FellowsEtAlTwo_reduction:adjacent:vertices} nor \cref{FellowsEtAlTwo_reduction:cliques:co:matching:new} can be applied
		and $G$ contains a vertex $v$ of degree three (where $N(v) = \{a, b, c\}$), then reduce $(G, k)$ to $(G', k)$ where $G'$ is the graph on vertex set $V(G) \setminus \{v\}$ and edge set $(E(G) \cap \binom{V(G')}{2}) \cup F$, where 
		\begin{align}\label{FellowsEtAlTwo_eq:reduction:degree:three:F}
			F \defeq  \left\lbrace \{a, b\}, \{b, c\} \right\rbrace \cup \left\lbrace \{a, x\} \mid x \in N_G(b) \right\rbrace {} & \cup \left\lbrace \{b, y\} \mid y \in N_G(c) \right\rbrace \\ 
			& \cup \left\lbrace \{c, z\} \mid z \in N_G(a) \right\rbrace. \nonumber
		\end{align}
	\end{reduction}
	
	We illustrate the above reduction rule in \cref{FellowsEtAlTwo_fig:reduction:degree:three}.
	
	\begin{figure}[t]
		\centering
		\includegraphics[height=.125\textheight]{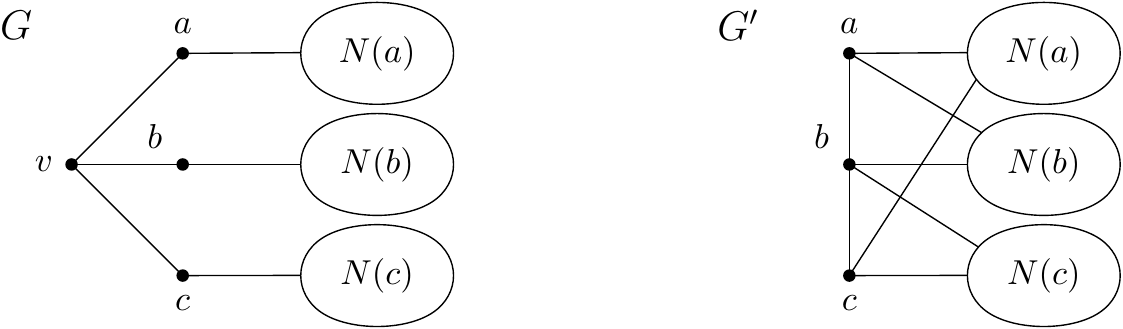}
		\caption{Illustration of \cref{FellowsEtAlTwo_reduction:degree:three}.}
		\label{FellowsEtAlTwo_fig:reduction:degree:three}
	\end{figure}
	
	\begin{proposition}
		\Cref{FellowsEtAlTwo_reduction:degree:three} is safe, i.e., if its conditions are satisfied, then $G$ contains a vertex cover of size $k$ if and only if $G'$ contains a vertex cover of size $k$.
	\end{proposition}
	\begin{proof}
		We first show that we can assume that there are no edges between the vertices in $N(v)$.
		\begin{claim*}
			If neither \cref{FellowsEtAlTwo_reduction:adjacent:vertices} nor \cref{FellowsEtAlTwo_reduction:cliques:co:matching:new} can be applied,
			then $N(v) = \{a, b, c\}$ is an independent set in $G$.
		\end{claim*}
		\begin{clproof}
			If $G[N(v)]$ contains at least two edges, then these two edges have a common endpoint, say $x \in N(v)$. But then, $N(v) \subseteq N[x]$, so we could have applied \cref{FellowsEtAlTwo_reduction:adjacent:vertices}, a contradiction. If $G[N(v)]$ contains precisely one edge, assume w.l.o.g.\ that $\{a, b\} \in E(G)$, then we could have applied \cref{FellowsEtAlTwo_reduction:cliques:co:matching:new} with $C_1 = \{a, b\}$, and $C_2 = \{c\}$. Clearly, $\{a, b\}$ and $\{c\}$ are cliques and $M \defeq \{\{a, c\}, \{b, c\}\}$, the set of non-edges of $G[C_1, C_2]$, satisfies the conditions of \cref{FellowsEtAlTwo_reduction:cliques:co:matching:new}\cref{FellowsEtAlTwo_reduction:cliques:co:matching:new:m}.
			\claimqed
		\end{clproof}
		
		Due to the previous claim, we will assume that $N(v) = \{a, b, c\}$ is an independent set throughout the following.
		
		\begin{claim}\label{FellowsEtAlTwo_claim:degree:three:cor:1}
			If $G$ has a vertex cover of size $k$, then $G'$ has a vertex cover of size at most $k$.		
		\end{claim}
		\begin{clproof}
			We first observe that, for each edge $e' \in E(G')$, either $e' \in E(G - v)$ or $e' \in F$. Hence, any vertex cover $X^*$ of $G - v$ is a vertex cover of $G'$ if each edge in $F$ has an endpoint in $X^*$, since by definition, $X^*$ contains an endpoint of each edge in $E(G - v)$.	
			
	Let $X$ be a vertex cover of $G$ of size $k$. If $v \notin X$, then $N(v) = \{a, b, c\} \subseteq X$. 
	By \cref{FellowsEtAlTwo_eq:reduction:degree:three:F}, each edge in $F$ has at least one endpoint in $\{a, b, c\}$ and hence in $X$, so we can conclude that $X$ is a vertex cover of $G'$ of size $k$. 
			
	Suppose $v \in X$ and note for the remainder of the proof that $X \setminus \{v\}$ is a vertex cover of $G-v$ of size $k-1$. We argue that we can assume that at most one vertex from $N(v)$ is contained in $X$: For the case that $N(v) \subseteq X$, we can apply the same argument as above to conclude that $X \setminus \{v\}$ is a vertex cover of $G'$ of size $k-1$.
	If $X$ contains precisely two vertices from $N(v) = \{a, b, c\}$, assume w.l.o.g.\ that $\{a, b\} \subseteq X$, then $X' \defeq X \setminus \{v\} \cup \{c\}$ is a vertex cover of $G'$ of size $k$, since again, $X'$ contains $N(v)$.
			
	We assume that $X$ contains at most one vertex from $N(v)$. If $X$ contains no vertex of $N(v)$, then $X$ must contain all of $N_G(a, b, c)$. Hence the only edges in $F$ that are not covered by $X$ -- see \cref{FellowsEtAlTwo_eq:reduction:degree:three:F} -- are incident with the vertex $b$. Together with the fact that $X \setminus \{v\}$ is a vertex cover of $G - v$, we can conclude that $X \setminus \{v\} \cup \{b\}$ is a vertex cover of $G'$.
	
	From now on, we assume that precisely one vertex of $N(v)$ is contained in the vertex cover $X$ of $G$. If $a \in X$, then $b, c \notin X$ and hence $N_G(b, c) \subseteq X$. Again, $X \setminus \{v\}$ is a vertex cover of $G - v$ and we observe that any edge in $e' \in F$ that does not have an endpoint in $X \setminus \{v\}$ is incident with the vertex $c$. By~\cref{FellowsEtAlTwo_eq:reduction:degree:three:F}, either $e' = \{b, c\}$ or $e' = \{c, z\}$ for some $z \in N(a)$. We can conclude that $X \setminus \{v\} \cup \{c\}$ is a vertex cover of $G'$. The remaining cases can be argued for similarly: If $b \in X$, then $X \setminus \{v\} \cup \{a\}$ is a vertex cover of $G'$ and if $c \in X$, then $X \setminus \{v\} \cup \{b\}$ is a vertex cover of $G'$.
	\claimqed
\end{clproof}
		
	\begin{claim}\label{FellowsEtAlTwo_claim:degree:three:cor:2}
		If $G'$ has a vertex cover of size $k$ then $G$ has a vertex cover of size $k$.
	\end{claim}
	\begin{clproof}	
	Throughout the following, let $X'$ be a vertex cover of $G'$ of size $k$. Since $\{a, b, c\}$ is \emph{not} an independent set in $G'$, we know that $X'$ has to contain at least one vertex of $\{a, b, c\}$. If $\{a, b, c\} = N(v) \subseteq X'$, then $X'$ contains an endpoint of each edge in $E(G) \setminus E(G') = \left\lbrace \{v, x\} \mid x \in \{a, b, c\} \right\rbrace$, so we can conclude that $X'$ is a vertex cover of $G$.
	
	We now consider the cases when $X'$ contains precisely two vertices from $\{a, b, c\}$. If $\{a, b\} \subseteq X'$ and hence $c \notin X'$, then $X'$ contains $N_G(a)$ as well, to cover the edges between the vertex $c$ and vertices in $N(a)$. It follows that $X' \setminus \{a\}$ is a vertex cover of $G - v$. Since each edge in $E(G) \setminus E(G - v)$ is incident with $v$, we can conclude that $X' \setminus \{a\} \cup \{v\}$ is a vertex cover of $G$. By similar arguments we have that if $X' \cap \{a, b, c\} = \{b, c\}$, then $X' \setminus \{b\} \cup \{v\}$ is a vertex cover of $G$ and if $X' \cap \{a, b, c\} = \{a, c\}$, then $X' \setminus \{c\} \cup \{v\}$ is a vertex cover of $G$.
	
	It remains to argue the case when $X'$ contains precisely one vertex from $\{a, b, c\}$. Note that the only possible such case is when this vertex is $b$. If $X'$ contained only the vertex $a$ (resp., $c$), then the edge $\{b, c\}$ (resp., $\{a, b\}$) would remain uncovered by $X'$. Suppose $X' \cap \{a, b, c\} = \{b\}$, so $a, c \notin X'$, implying that $N(a, b, c) \subseteq X'$. Hence, $X' \setminus \{b\}$ is a vertex cover of $G - v$ and $X' \setminus \{b\} \cup \{v\}$ is a vertex cover of $G$.
	\claimqed
	\end{clproof}
	
	In the light of \cref{FellowsEtAlTwo_claim:degree:three:cor:1,FellowsEtAlTwo_claim:degree:three:cor:2}, the proposition is proved.
  \qed
\end{proof}

	The next reduction rule kernelizes all vertices that have degree four and whose neighborhood induces a subgraph with more than two edges.
	\begin{figure}[t]
		\centering
		\includegraphics[height=.125\textheight]{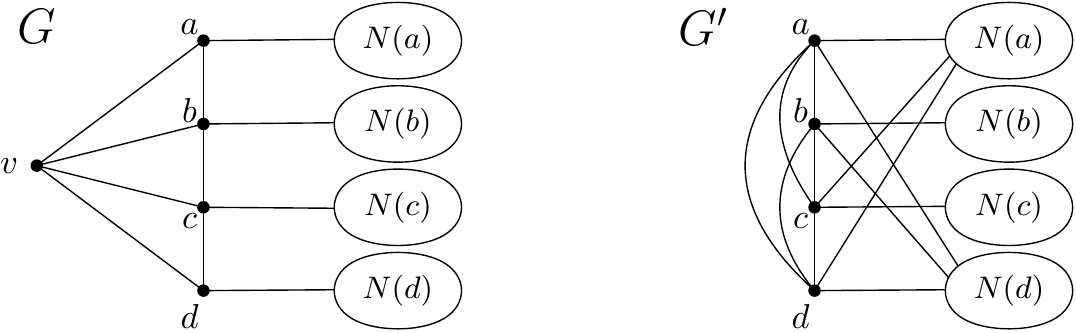}
		\caption{Illustration of \cref{FellowsEtAlTwo_reduction:degree:four}.}
		\label{FellowsEtAlTwo_fig:reduction:degree:four}
	\end{figure}
	
	\begin{reduction}\label{FellowsEtAlTwo_reduction:degree:four}
		If neither \cref{FellowsEtAlTwo_reduction:adjacent:vertices} nor \cref{FellowsEtAlTwo_reduction:cliques:co:matching:new} can be applied and
		$G$ contains a vertex $v$ with degree four such that $G[N(v)]$ has at least three edges, 
		then we can assume that (up to renaming the vertices in $N(v) = \{a, b, c, d\}$) $G[N(v)]$ is an $(a, d)$-path. We reduce $(G, k)$ to $(G', k)$, where $G'$ is the graph on vertex set $V(G) \setminus \{v\}$ and edge set $(E(G) \cap \binom{V(G')}{2}) \cup F$, where
		\begin{align}\label{FellowsEtAlTwo_eq:reduction:degree:four:path}
			F \defeq \binom{N(v)}{2} {} &\cup \left\lbrace \{x, y\} \mid x \in \{a, b\}, y \in N(d) \right\rbrace 
			                            \cup \left\lbrace \{x, y\} \mid x \in \{c, d\}, y \in N(a) \right\rbrace.
		\end{align}	
	\end{reduction}
	
	We illustrate \cref{FellowsEtAlTwo_reduction:degree:four} in \cref{FellowsEtAlTwo_fig:reduction:degree:four}.
	\begin{proposition}
		\Cref{FellowsEtAlTwo_reduction:degree:four} is safe, i.e., if its conditions are satisfied, then $G[N(v)]$ is a path and $G$ has a vertex cover of size $k$ if and only if $G'$ has a vertex cover of size $k$.
	\end{proposition}
	\begin{proof}
		We first justify the assumption that $G[N(v)]$ induces a path.
		\begin{claim}\label{FellowsEtAlTwo_claim:degree:four:conditions}
			If neither \cref{FellowsEtAlTwo_reduction:adjacent:vertices} nor \cref{FellowsEtAlTwo_reduction:cliques:co:matching:new} can be applied
			and $G$ and $v$ are as above,
			then $G[N(v)]$ induces a path.
		\end{claim}
		\begin{clproof}
			Suppose not. If $G[N(v)]$ contains at least five edges, then there has to be a vertex $x \in \{a, b, c, d\}$ which is incident with three of these edges. Hence, $N(v) \subseteq N[x]$ and we could have applied \cref{FellowsEtAlTwo_reduction:adjacent:vertices}, a contradiction. 
			
			Suppose there are four edges in $G[N(v)]$. There are only two non-isomorphic graphs on four vertices and four edges, see the right hand side of \cref{FellowsEtAlTwo_fig:graphs:on:four:vertices}. If $G[N(v)]$ induces a $C_4$, assume w.l.o.g.\ that its vertices appear in the order $a-b-c-d-a$, then we can partition $N_G(v)$ into cliques $C_1 = \{a, b\}$ and $C_2 = \{c, d\}$, and we observe that the set $M \defeq \{\{a, c\}, \{b, d\}\}$ of non-edges of $G[C_1, C_2]$ satisfies the condition of \cref{FellowsEtAlTwo_reduction:cliques:co:matching:new}\cref{FellowsEtAlTwo_reduction:cliques:co:matching:new:m}: For $x \in \{a, b\}$ there is precisely one element in $M$ that contains $x$. We could have applied \cref{FellowsEtAlTwo_reduction:cliques:co:matching:new}, a contradiction.
			If $G[N(v)]$ induces an $H_4$, assume w.l.o.g.\ that the $3$-cycle is $a-b-c-a$ and $d$ is the pendant vertex adjacent to $c$, 
			then we have that $N(v) \subseteq N[c]$ and we could have applied \cref{FellowsEtAlTwo_reduction:adjacent:vertices}.
			
			We can assume that $G[N(v)]$ contains precisely three edges. There are three pairwise non-isomorphic graphs on four vertices and three edges, shown in the left-hand side of \cref{FellowsEtAlTwo_fig:graphs:on:four:vertices}. If $G[N(v)]$ induces a star ($S_4$) with center $a$, then we have that $N(v) \subseteq N[a]$, so we could have applied \cref{FellowsEtAlTwo_reduction:adjacent:vertices}, a contradiction.
			If $G[N(v)]$ induces a $K_3 \cup K_1$, then let $C_1 \defeq \{a, b, c\}$ be the vertices that induce a $K_3$ and $C_2 = \{d\}$, where $d$ is the remaining vertex of $N(v)$. Clearly, $C_1$ and $C_2$ induce cliques and the set of non-edges of $G[C_1, C_2]$ is such that it contains precisely one element incident with each vertex in $C_1$. Hence, we could have applied \cref{FellowsEtAlTwo_reduction:cliques:co:matching:new}, a contradiction.
			We can conclude that the only case that has not been covered is when $G[N(v)]$ induces a $P_4$, which proves the claim.
			\claimqed
		\end{clproof}
		\begin{figure}[t]
		\centering
		\includegraphics[height=.1\textheight]{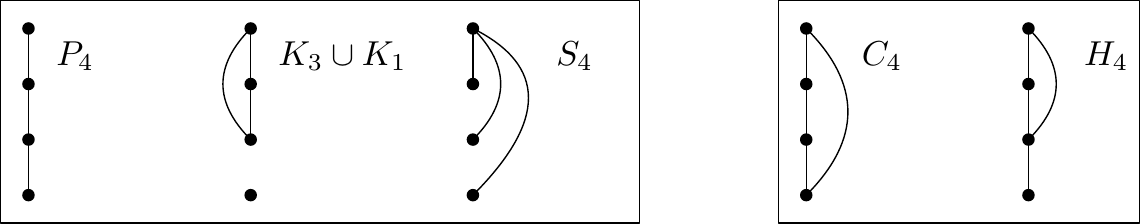}
		\caption{In the left box, all pairwise non-isomorphic graphs on four vertices and three edges and in the right box, all pairwise non-isomorphic graphs on four vertices and four edges are shown.}
		\label{FellowsEtAlTwo_fig:graphs:on:four:vertices}
	\end{figure}		
	\begin{claim}\label{FellowsEtAlTwo_claim:degree:four:safe}
			$G$ contains a vertex cover of size $k$ if and only if $G'$ contains a vertex cover of size $k$.
		\end{claim}
		\begin{clproof}
	($\Rightarrow$) Suppose $G$ has a vertex cover $X$ of size $k$. If $N(v) \subseteq X$, then every edge in $E(G') \setminus E(G)$ has at least one endpoint in $X$ by construction, so $X$ is a vertex cover of $G'$ as well. If $X$ contains precisely three vertices of $N(v)$, suppose w.l.o.g. that $X \cap N(v) = \{a, b, c\}$, then $v$ has to be contained in $X$ as well, otherwise the edge $\{v, d\}$ remains uncovered. Hence, $X \setminus \{v\} \cup \{d\} \supseteq N(v)$ is a vertex cover of $G'$ of size $k$.
		
		From now on suppose that $X$ contains precisely two vertices from $N(v)$ and note that in all of the following cases, $v \in X$. The only vertex covers of $G[N(v)]$ of size two are $\{b, c\}$, $\{a, c\}$ and $\{b, d\}$. First observe that any triple of vertices from $N_G(v)$ covers the edges $\binom{N(v)}{2}$, so we do not have to consider them explicitly in the following discussion.
		
		Suppose $X \cap N(v) = \{b, c\}$. Then, $N(a, d) \subseteq X$ , so all edges in $E(G') \setminus E(G) \setminus \binom{N(v)}{2}$ are covered by $X$. We can conclude that $X \setminus \{v\} \cup \{a\}$ (also $X \setminus \{v\} \cup \{d\}$) is a vertex cover of $G'$. If $X \cap N(v) = \{a, c\}$, then $N(b, d) \subseteq X$ and all edges in $E(G') \setminus E(G) \setminus \binom{N(v)}{2}$ that are not covered by $X$ are between $d$ and $N(a)$, so $X \setminus \{v\} \cup \{d\}$ is a vertex cover of $G'$. Similarly, if $X \cap N(v) = \{b, d\}$, then $X \setminus \{v\} \cup \{a\}$ is a vertex cover of $G'$.
		Since $G[N(v)]$ does not have a vertex cover of size at most $1$, this concludes the proof of the first direction.
		
		($\Leftarrow$) Since $N_G(v)$ is a clique in $G'$, any vertex cover of $G'$ contains at least three vertices from $N_G(v)$. Let $X'$ be a vertex cover of $G'$ of size $k$. If $N_G(v) \subseteq X'$, then $X'$ is also a vertex cover of $G$ since each edge in $E(G) \setminus E(G')$ has an endpoint in $N_G(v)$. In the remainder, we can assume that $X'$ contains precisely three vertices from $N_G(v)$. Suppose $X' \cap N_G(v) = \{a, b, c\}$. Since $d \notin X'$, we have that $N(a, d) \subseteq X'$, hence all edges between $a$ and $N(a)$ are covered by $N(a) \subseteq X'$. Together with the observation that $\{b, c\}$ is a vertex cover of $G[N(v)]$, we can conclude that $X' \setminus \{a\} \cup \{v\}$ is a vertex cover of $G$. If $X' \cap N_G(v) = \{a, b, d\}$, then $N(a, c) \subseteq X'$ since $c \notin X'$. By the same reasoning as before, we can observe that $X' \setminus \{a\} \cup \{v\}$ is a vertex cover of $G$. Similarly, if $X' \cap N_G(v) = \{b, c, d\}$ or if $X' \cap N_G(v) = \{a, c, d\}$ then we can argue that $X' \setminus \{d\} \cup \{v\}$ is a vertex cover of $G$.
		\claimqed
		\end{clproof}
		Now, by \cref{FellowsEtAlTwo_claim:degree:four:conditions}, we know that under the conditions stated in \cref{FellowsEtAlTwo_reduction:degree:four}, $G[N(v)]$ induces a path which together with \cref{FellowsEtAlTwo_claim:degree:four:safe} proves the proposition.
		\qed
	\end{proof}
	It is easy to see that \cref{FellowsEtAlTwo_reduction:pendant:edge,FellowsEtAlTwo_reduction:degree:two,FellowsEtAlTwo_reduction:adjacent:vertices,FellowsEtAlTwo_reduction:degree:three,FellowsEtAlTwo_reduction:degree:four} can be executed in polynomial time. We observe (naively) that \cref{FellowsEtAlTwo_reduction:cliques:co:matching:new} can be executed in time $\cO\left(n \cdot 2^\alpha \cdot \degreev^{\cO(1)}\right)$, where $\degreev$ denotes the degree of the vertex $v$. 
	Since for our purposes, $\degreev \le 4$ is sufficient, \cref{FellowsEtAlTwo_reduction:cliques:co:matching:new} runs in polynomial time as well and we have the following theorem. (Note that none of the presented reductions increases the parameter value.)
	\begin{theorem}[cf.\ Fellows and Stege \cite{FellowsEtAlTwo_FS99}]\label{FellowsEtAlTwo_thm:kernelization:degree}
		There is a polynomial-time algorithm that given an instance $(G, k)$ of \vc ~outputs an equivalent instance $(G', k')$, where $k' \le k$,
		\begin{enumerate}[label=(\Roman*)]
			\item the minimum degree of $G'$ is at least four and
			\item for all vertices $v \in V(G')$ with $\deg_{G'}(v) = 4$, $G'[N(v)]$ contains at most two edges.
		\end{enumerate}
	\end{theorem}
	
	\section{Automated Vertex Cover Kernelization}\label{FellowsEtAlTwo_sec:automated}
Many reduction rules for various problems are instances of the same class that can be described as
``find a subgraph $H$ with boundary $X$ and replace it with a graph $H^*$''.
For example, \cref{FellowsEtAlTwo_reduction:degree:two} can be formulated as
``find a $P_3$ $(a,v,b)$ with boundary $\{a,b\}$ (a kind of ``local surgery'') and replace it with $z_{N[v]}$''.
Note however, that boundary connections might change during the replacement.
Proposed reduction rules of this type can be checked in roughly $2^{|X|}$ times the time it takes
to compute a minimum vertex cover in $H$ and $H^*$.
\newcommand{\nakedG}[1]{
  \foreach \i/\x/\y in {0/-4/-2, 1/-4/3, 2/-2/0, 3/0/3, 4/0/1, 5/0/-1, 6/0/-3}
    \node[smallvertex] (u\i) at (\x,\y) {};
  \foreach \i/\x/\y/\a in {0/2/-3/-135, 1/2/0/45, 2/2/2/45}{
    \ifx\relax#1\relax
      \node[smallvertex, fill=FellowsEtAlTwo_gray] (x\i) at (\x,\y) {};
    \else
      \node[smallvertex, fill=FellowsEtAlTwo_gray, label=\a:{\small $x_\i$}] (x\i) at (\x,\y) {};
    \fi
  }
  \foreach \i/\j in {0/1, 0/2, 0/6, 1/2, 1/4, 1/3, 2/4, 2/5, 3/4}
    \draw (u\i) -- (u\j);
  \foreach \i/\j in {2/1, 3/2, 4/1, 4/2, 5/0, 5/1, 6/0}
    \draw (u\i) -- (x\j);
  \begin{pgfonlayer}{background}
    \draw[FellowsEtAlTwo_gray, rounded corners, fill=lightgray!80!white] (-5,-4) rectangle (2,4);
    \node at (-4.2,-3.2) {$G$};
  \end{pgfonlayer}
}
\newcommand{\nakedH}{
  \foreach \i in {0,1,2} \node[smallvertex] (v\i) at (4,2*\i-2) {};
  \foreach \u/\v in {x0/v0, x0/v1, x1/v1, x2/v1, x2/v2, v0/v1, v1/v2} \draw (\u) -- (\v);
  \draw (x0) to[bend right] (x1);
  \begin{pgfonlayer}{background}
    \draw[FellowsEtAlTwo_gray, rounded corners, fill=lightgray!30!white] (2,-4) rectangle (5,4);
    \node at (4,-3.2) {$H$};
  \end{pgfonlayer}
}
\newcommand{\nakedHprime}{
  \node[smallvertex] (v0) at (4,-2) {};
  \node[smallvertex] (v1) at (4,2) {};
  \foreach \u/\v/\b in {x0/v0/0, x2/v1/0, x0/x1/30, x1/x2/30} \draw (\u) to[bend right=\b] (\v);
  \begin{pgfonlayer}{background}
    \draw[FellowsEtAlTwo_gray, rounded corners, fill=lightgray!30!white] (2,-4) rectangle (5,4);
    \node at (4,-3.2) {$H'$};
  \end{pgfonlayer}
}
\begin{figure}
  \centering
  \begin{tabular}[h]{ccc}
    \multirow{2}{*}[9mm]{
      \begin{subfigure}{5.3cm}
        \begin{tikzpicture}[scale=.5, label distance=-2pt]
          \nakedG{1}
          \nakedH{}
        \end{tikzpicture}
				\subcaption{}
        \label{FellowsEtAlTwo_fig:bigG}
      \end{subfigure}
    }
  &
    \begin{subfigure}{3cm}
      \begin{tikzpicture}[scale=.25, label distance=-2pt]
        \nakedG{}
        \nakedH{}
        \foreach \u in {u0, u2, u3, u4, x0, x1, x2, v1} \node[smallvertex, fill=black] at (\u) {};
      \end{tikzpicture}
			\subcaption{}
      \label{FellowsEtAlTwo_fig:Hprofile2}
    \end{subfigure}
  &
    \begin{subfigure}{3cm}
      \begin{tikzpicture}[scale=.25, label distance=-2pt]
        \nakedG{}
        \nakedH{}
        \foreach \u in {u1, u2, u3, u4, u5, u6, x0, x2, v1} \node[smallvertex, fill=black] at (\u) {};
      \end{tikzpicture}
			\subcaption{}
      \label{FellowsEtAlTwo_fig:Hprofile012}
    \end{subfigure}
  \\ &&\\
  &
    \begin{subfigure}{3cm}
      \begin{tikzpicture}[scale=.25, label distance=-2pt]
        \nakedG{}
        \nakedHprime{}
        \foreach \u in {u0, u2, u3, u4, x0, x1, x2} \node[smallvertex, fill=black] at (\u) {};
      \end{tikzpicture}
			\subcaption{}
      \label{FellowsEtAlTwo_fig:H'profile2}
    \end{subfigure}
  &
    \begin{subfigure}{3cm}
      \begin{tikzpicture}[scale=.25, label distance=-2pt]
        \nakedG{}
        \nakedHprime{}
        \foreach \u in {u1, u2, u3, u4, u5, u6, x0, x2} \node[smallvertex, fill=black] at (\u) {};
      \end{tikzpicture}
			\subcaption{}
      \label{FellowsEtAlTwo_fig:H'profile012}
    \end{subfigure}
  \end{tabular}
  \caption[Illustration of notation.]{Illustration of notation.
    \textbf{(a)} shows a graph $\glue{G}{H}$ resulting from gluing $G$ (darker background) and $H$ (lighter background).
    \textbf{(b)} shows a vertex cover of $\glue{G}{H}$ compatible with $\{2\}$,
      minimizing the intersection with $V(H)$ and, thus, implying $\profile{H}^G(\{2\})=4$.
    \textbf{(c)} shows a vertex cover~$S$ of $\glue{G}{H}$ compatible with $\{0,1,2\}$,
      minimizing the intersection with $V(H)$ and, thus, implying $\profile{H}^G(\{0,1,2\})=3$.
      Note that $S$ is not necessarily minimum or even minimal for $\glue{G}{H}$.
    \textbf{(d) and (e)} show the same as (b) and (c) but with the graph $H'$ resulting from the application of \cref{FellowsEtAlTwo_reduction:adjacent:vertices}
      to a vertex in $H$, proving $\profile{H'}^G(X)=\profile{H}^G(X)-1$ for $X\in\{\{2\}, \{0,1,2\}\}$.
  }
  \label{FellowsEtAlTwo_fig:profile}
\end{figure}
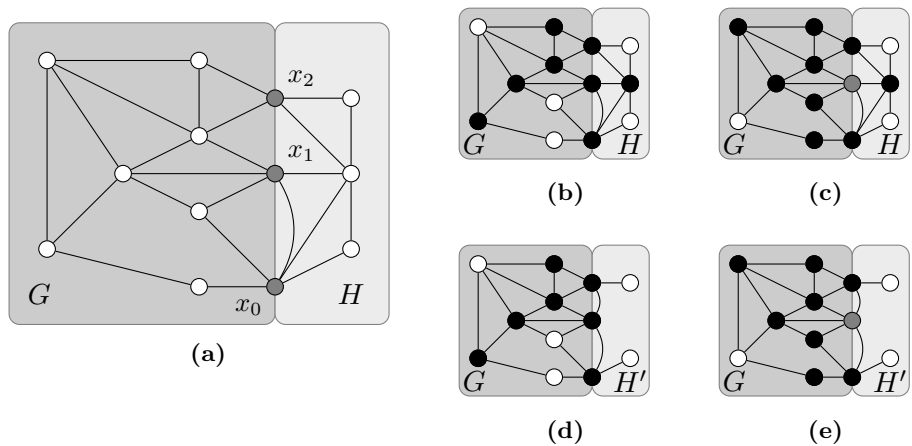
\paragraph{Notation.}
See \cref{FellowsEtAlTwo_fig:profile} for an illustration.
Let $G$ be a graph in which $t$ non-isolated vertices $x_1,\ldots,x_t$ are bijectively labeled with the integers in $[t]$.
Then, we call $G$ \emph{$t$-boundaried}.
If $\{x_1,\ldots,x_t\}$ is an independent set in $G$, we call $G$ \emph{strongly $t$-boundaried}.
For a $t$-boundaried graph $H$ and a strongly $t$-boundaried graph $G$,
we let $\glue{G}{H}$ denote the result of \emph{gluing} $G$ and $H$, that is,
identifying the vertices with the same label in the disjoint union of $G$ and $H$.
We call a set $S\subseteq V(G)$ \emph{compatible} with a set $X\subseteq [t]$ in $G$ if 
$N_G(x_i)\subseteq S \iff i\in X$ for all $i\in [t]$.
Let $\profile{H}^G\colon 2^{[t]}\to\N$ be such that, for all $X\subseteq [t]$,
$\profile{H}^G(X)$ is the smallest number of vertices of $H$
contained in any vertex cover $S$ of $\glue{G}{H}$ that is compatible with $X$ in $G$
(and $\infty$ if no such $S$ exists).
Then, we call $\profile{H}^G$ the \emph{profile} of $H$ in $G$.
It turns out that the profile is indeed independent of $G$, so we drop the superscript.

\begin{lemma}\label{FellowsEtAlTwo_lem:profile equal}
  Let $G$, $G'$ be strongly $t$-boundaried,
  let $H$ be $t$-boundaried and
  let $X\subseteq [t]$.
  Then, $\profile{H}^G(X)=\profile{H}^{G'}(X)$ for all $X\subseteq [t]$.
\end{lemma}
\begin{proof}
  Let $X\subseteq [t]$ and
  let $S$ and $S'$ be vertex covers of $\glue{G}{H}$ and $\glue{G'}{H}$ that are compatible with $X$ in $G$ and $G'$, respectively,
  such that $|S\cap V(H)|$ and $|S'\cap V(H)|$ are minimum among all such vertex covers.
  To prove the claim, we show that $S^* \defeq (S\setminus V(H))\cup(S'\cap V(H))$ is a vertex cover of $\glue{G}{H}$.
  By symmetry, the same follows for $S$ and $S'$ inversed, which then implies the lemma.
  Towards a contradiction, assume that $\glue{G}{H}$ contains an edge $uv$ such that $u,v\notin S^*$.
  Then, exactly one of $u$ and $v$ is in $H$ as, otherwise, $S$ is not a vertex cover of $\glue{G}{H}$.
  Without loss of generality, let $u\in V(H)$ and $v\notin V(H)$.
  Thus, $u$ is a boundary vertex $x_i$ and $v\notin S$.
  Since $S$ and $S'$ are compatible with $X$ in $G$ and $G'$, respectively,
  we know that $N_G(x_i)\subseteq S \iff i\in X \iff N_{G'}(x_i)\subseteq S'$ and,
  as $v\in N_G(u)\setminus S$, we have $N_{G'}(u)\nsubseteq S'$.
  However, since $S'$ is a vertex cover of $\glue{G'}{H}$, we have $u\in S'$, which contradicts $u\notin S^*$ since $u\in V(H)$.
  \qed
\end{proof}

We observe that two $t$-boundaried graphs $H$ and $H'$ with the same profile
can be swapped for one another in any graph $G$ without changing the size of an optimal vertex cover,
that is, $\glue{G}{H}$ and $\glue{G}{H'}$ have the same vertex cover number.
More generally, for any $c\in\N$, we say that $H$ and $H'$ are \emph{$c$-equivalent} if $\forall_{X\subseteq [t]}\;\profile{H}(X)=\profile{H'}(X)+c$.
In this way, for any fixed size $t$, the profile gives rise to an equivalence relation on the set of $t$-boundaried graphs.
This relation allows automated discovery of reduction rules that remove vertices with undesirable properties from the input graph.
The idea is, for each induced subgraph $H$ having an undesirable property~$\Pi$,
to replace $H$ by some $c$-equivalent $H'$ that does not suffer from $\Pi$, 
while reducing $k$ by $c$.

\begin{lemma}\label{FellowsEtAlTwo_lem:equiv}
  Let $G$ be strongly $t$-boundaried,
  let $H$ and $H'$ be $t$-boundaried and $c$-equivalent for some $c\in\N$, and
  let $k\in\N$.
  Then, $\glue{G}{H}$ has a vertex cover of size at most $k$
  if and only if
  $\glue{G}{H'}$ has a vertex cover of size at most $k-c$.
\end{lemma}
\begin{proof}
  As ``$\Rightarrow$'' is completely analogous to ``$\Leftarrow$'', we only prove the latter.
  To this end,
  let $S$ be a smallest vertex cover of $\glue{G}{H}$ that, among all such vertex covers, minimizes $|S\cap V(H)|$.
  Let $X \defeq \{i \mid N_G(x_i)\subseteq S\}$ and note that $S$ is compatible with $X$ in $G$ and $|S\cap V(H)|=\profile{H}(X)$.
  Let $S'$ be a smallest vertex cover of $\glue{G}{H'}$ that is compatible with $X$ in $G$ and,
  among all such vertex covers, minimizes $|S'\cap V(H')|$.
  As $H$ and $H'$ are $c$-equivalent, we know that $|S\cap V(H)|=|S'\cap V(H')|+c$.
  We show that $S^* \defeq (S\setminus V(H))\cup(S'\cap V(H))$ is a vertex cover of $\glue{G}{H'}$
  (clearly, $|S^*|\leq|S\setminus V(H)|+|S'\cap V(H')|= |S\setminus V(H)| + |S\cap V(H)| + c = |S|+c$).
  Towards a contradiction, assume that there is an edge $uv$ of $\glue{G}{H'}$ with $u,v\notin S^*$.
  If $u,v\notin V(H')$, then $S$ is not a vertex cover of $\glue{G}{H}$ and,
  if $u,v\in V(H')$, then $S'$ is not a vertex cover of $\glue{G}{H'}$.
  Thus, without loss of generality, $u\in V(H')$ and $v\notin V(H')$,
  implying that $u$ is a boundary vertex $x_i$.
  Since $u\notin S^*$, we know that $u\notin S'$ and, since $S'$ is a vertex cover of $\glue{G}{H'}$, we have $N_G(u)\subseteq S'$.
  Since $S'$ is compatible with $X$, we have $i\in X$ and, since $S$ is compatible with $X$, we have $N_G(u)\subseteq S$,
  implying $v\in S$ which contradicts $v\notin S^*$ since $v\notin V(H)$.
	\qed
\end{proof}
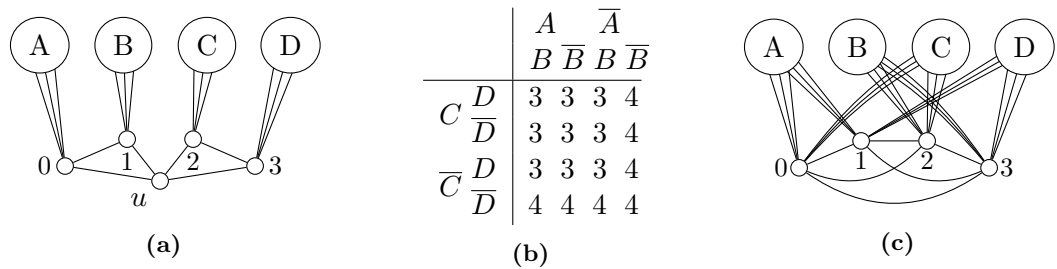
\begin{figure}[t]
  \hfill
  \begin{subfigure}{.38\textwidth}
  	\centering
    \begin{tikzpicture}[label distance=-2pt, yscale=.6, xscale=.5, baseline={(0,3)}]
      \node[smallvertex, label=below left:$u$] (u) at (0,4) {};
      \foreach \i/\N in {0/A, 1/B, 2/C, 3/D}{
        \foreach \j in {0,1,2}{
          \coordinate (N\N\j) at ($(0,4)+(90:3)+(\i*2.2-3.5,0)+(\j/3,0)$);
        }
      }
      \foreach \i/\a/\N in {0/180/A, 1/-90/B, 2/-90/C, 3/0/D}{
        \node[smallvertex, label=\a:{\small $\i$}] (v\i) at (120 - 20*\i:5) {} edge (u);
        \foreach \j in {0,1,2} \draw (v\i) -- (N\N\j);
        \node[draw, ellipse, fill=white] at (N\N1) {\N};
      }
      \draw (v0) -- (v1);
      \draw (v2) -- (v3);
    \end{tikzpicture}
		\subcaption{}
    \label{FellowsEtAlTwo_fig:d4 auto orig}
  \end{subfigure}
  \hfill
  \begin{subfigure}{.2\textwidth}
  	\centering
    \begin{tabulary}{25mm}{Cc|CCCC}
      && \multicolumn{2}{c}{$A$} & \multicolumn{2}{c}{$\overline{A}$} \\
      && $B$ & $\overline{B}$ & $B$ & $\overline{B}$ \\\hline
      \multirow{2}{*}{$C$} & $D$            & 3 & 3 & 3 & 4   \\
      & $\overline{D}$                      & 3 & 3 & 3 & 4   \\
      \multirow{2}{*}{$\overline{C}$} & $D$ & 3 & 3 & 3 & 4   \\
      & $\overline{D}$                      & 4 & 4 & 4 & 4   \\
    \end{tabulary}
		\subcaption{}
    \label{FellowsEtAlTwo_fig:d4 auto profile}
  \end{subfigure}
  \hfill
  \begin{subfigure}{.38\textwidth}
  	\centering
    \begin{tikzpicture}[label distance=-3pt, yscale=.6, xscale=.5, baseline={(0,3)}]
      \foreach \i/\N in {0/A, 1/B, 2/C, 3/D}{
        \foreach \j in {0,1,2}{
          \coordinate (N\N\j) at ($(0,4)+(90:3)+(\i*2.2-3.5,0)+(\j/3,0)$);
        }
      }
      \foreach \i/\a/\N in {0/180/A, 1/-90/B, 2/-90/C, 3/0/D}{
        \node[smallvertex, label=\a:{\small $\i$}] (v\i) at (120 - 20*\i:5) {};
        \node[draw, ellipse, fill=white] at (N\N1) {\N};
      }
      \begin{pgfonlayer}{background}
        \foreach \k in {0,1,2}
          \foreach \i/\j/\b in {0/A/0, 1/D/0, 2/C/0, 3/D/0, 1/A/0, 2/B/0, 0/C/10, 3/B/-10} \draw (v\i) to[bend left=\b] (N\j\k);
        \foreach \i/\j/\b in {0/1/0, 2/3/0, 1/2/0, 0/2/30, 1/3/30, 0/3/30} \draw (v\i) to[bend right=\b] (v\j);
      \end{pgfonlayer}
    \end{tikzpicture}
		\subcaption{}
    \label{FellowsEtAlTwo_fig:d4 auto new}
  \end{subfigure}
  \hfill
  \caption[Illustration of \cref{FellowsEtAlTwo_rr:deg4 auto}.]{Illustration of \cref{FellowsEtAlTwo_rr:deg4 auto}.
    \textbf{(a)} shows the degree-four vertex $u$ with two edges in its neighborhood. $A$--$D$ represent the sets of neighbors of $0$--$3$, respectively, in the rest of the graph ($A$--$D$ may mutually intersect or be empty).
    \textbf{(b)} shows the profile of (a) where the entry $3$ at position $(A,\overline{B},C,D)$ means that three vertices are needed to cover all edges of (a), assuming all vertices in $A$, $C$, and $D$ are already in the cover. Indeed, $\{u,1,2\}$ is a size-3 cover in this case.
    \textbf{(c)} shows a subgraph that is 0-equivalent to (a), that is, (b) is also the profile of~(c).
  }
  \label{FellowsEtAlTwo_fig:deg4 auto}
\end{figure}

Given a $t$-boundaried graph $H$ and a property $\Pi$,
we can enumerate all $t$-boundaried graphs $H'$ that are $c$-equivalent to $H$ for some $c$ and
that do not suffer from $\Pi$.

\paragraph{Two Examples.}
A proof-of-concept implementation\footnote{\url{https://github.com/igel-kun/VC_min_deg}}
was used to attack the remaining cases of degree-four vertices (see \cref{FellowsEtAlTwo_sec:small:degree}).
For a given $t$-boundaried graph~$H$ or profile~$\profile{H}$ and a given number~$n$,
the implementation enumerates all strongly $t$-boundaried, $n$-vertex graphs $H'$
and outputs $H'$ if $\profile{H}(X)=\profile{H'}(X)$ for all $X\subseteq [t]$.
Feeding the graphs displayed in \cref{FellowsEtAlTwo_fig:d4 auto orig,FellowsEtAlTwo_fig:d4 auto2 orig}, the implementation yielded, in 5s and 6s, respectively,
reduction rules that remove degree-four vertices whose neighborhood contains exactly two edges.
\begin{reduction}\label{FellowsEtAlTwo_rr:deg4 auto}
  Let $G$ contain the 4-boundaried graph $H$ depicted in \cref{FellowsEtAlTwo_fig:d4 auto orig} as an induced subgraph.
  Then, replace $H$ by the 4-boundaried graph $H'$ depicted in \cref{FellowsEtAlTwo_fig:d4 auto new}.
\end{reduction}
\begin{figure}[t]
  \hfill
  \begin{subfigure}{.38\textwidth}
  	\centering
    \begin{tikzpicture}[label distance=-2pt, yscale=.6, xscale=.5, baseline={(0,3)}]
      \node[smallvertex, label=below left:$u$] (u) at (0,4) {};
      \foreach \i/\N in {0/A, 1/B, 2/C, 3/D}{
        \foreach \j in {0,1,2}{
          \coordinate (N\N\j) at ($(0,4)+(90:3)+(\i*2.2-3.5,0)+(\j/3,0)$);
        }
      }
      \foreach \i/\a/\N in {0/180/A, 1/-90/B, 2/-90/C, 3/0/D}{
        \node[smallvertex, label=\a:{\small $\i$}] (v\i) at (120 - 20*\i:5) {} edge (u);
        \foreach \j in {0,1,2} \draw (v\i) -- (N\N\j);
        \node[draw, ellipse, fill=white] at (N\N1) {\N};
      }
      \draw (v0) -- (v1);
      \draw (v1) -- (v2);
    \end{tikzpicture}
		\subcaption{}
    \label{FellowsEtAlTwo_fig:d4 auto2 orig}
  \end{subfigure}
  \hfill
  \begin{subfigure}{.2\textwidth}
  	\centering
    \begin{tabulary}{25mm}{Cc|CCCC}
      && \multicolumn{2}{c}{$A$} & \multicolumn{2}{c}{$\overline{A}$} \\
      && $B$ & $\overline{B}$ & $B$ & $\overline{B}$ \\\hline
      \multirow{2}{*}{$C$} & $D$            & 2 & 2 & 3 & 3   \\
      & $\overline{D}$                      & 3 & 3 & 4 & 4   \\
      \multirow{2}{*}{$\overline{C}$} & $D$ & 3 & 3 & 3 & 4   \\
      & $\overline{D}$                      & 4 & 4 & 4 & 4   \\
    \end{tabulary}
		\subcaption{}
    \label{FellowsEtAlTwo_fig:d4 auto2 profile}
  \end{subfigure}
  \hfill
  \begin{subfigure}{.38\textwidth}
  	\centering
    \begin{tikzpicture}[label distance=-3pt, yscale=.6, xscale=.5, baseline={(0,3)}]
      \foreach \i/\N in {0/A, 1/B, 2/C, 3/D}{
        \foreach \j in {0,1,2}{
          \coordinate (N\N\j) at ($(0,4)+(90:3)+(\i*2.2-3.5,0)+(\j/3,0)$);
        }
      }
      \foreach \i/\a/\N in {0/180/A, 1/-90/B, 2/-90/C, 3/0/D}{
        \node[smallvertex, label=\a:{\small $\i$}] (v\i) at (120 - 20*\i:5) {};
        \node[draw, ellipse, fill=white] at (N\N1) {\N};
      }
      \begin{pgfonlayer}{background}
        \foreach \i/\j/\b in {0/A/0, 0/C/10, 2/C/0, 2/D/0, 1/B/0, 1/D/10, 3/A/-10, 3/D/0}
          \foreach \k in {0,1,2} \draw (v\i) to[bend left=\b] (N\j\k);
        
        \foreach \i/\j/\b in {2/3/0, 1/2/0, 1/3/30} \draw (v\i) to[bend right=\b] (v\j);
      \end{pgfonlayer}
    \end{tikzpicture}
		\subcaption{}
    \label{FellowsEtAlTwo_fig:d4 auto2 new}
  \end{subfigure}
  \hfill
  \caption[Illustration of \cref{FellowsEtAlTwo_rr:deg4 auto2}.]{Illustration of \cref{FellowsEtAlTwo_rr:deg4 auto2}.
    \textbf{(a)} shows the degree-four vertex $u$ with two edges in its neighborhood. $A$--$D$ represent the sets of neighbors of $0$--$3$, respectively, in the rest of the graph.
    \textbf{(b)} shows the profile of (a) where the entry $2$ at position $(A,\overline{B},C,D)$ means that two vertices are needed to cover all edges of (a), assuming all vertices in $A$, $C$, and $D$ are already in the cover. Indeed, $\{u,1\}$ is a size-2 cover in this case.
    \textbf{(c)} shows a subgraph that is 0-equivalent to (a), that is, (b) is also the profile of~(c).
  }
  \label{FellowsEtAlTwo_fig:deg4 auto2}
\end{figure}
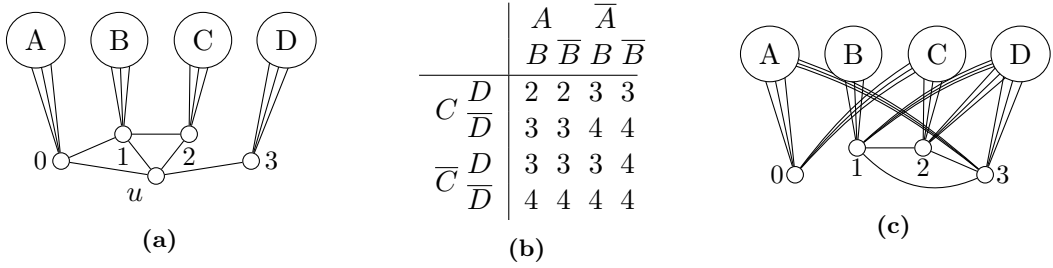
\begin{reduction}\label{FellowsEtAlTwo_rr:deg4 auto2}
  Let $G$ contain the 4-boundaried graph $H$ depicted in \cref{FellowsEtAlTwo_fig:d4 auto2 orig} as an induced subgraph.
  Then, replace $H$ by the 4-boundaried graph $H'$ depicted in \cref{FellowsEtAlTwo_fig:d4 auto2 new}.
\end{reduction}
By \cref{FellowsEtAlTwo_lem:equiv}, correctness of \cref{FellowsEtAlTwo_rr:deg4 auto,FellowsEtAlTwo_rr:deg4 auto2} can be verified by
convincing oneself that the profiles in \cref{FellowsEtAlTwo_fig:deg4 auto,FellowsEtAlTwo_fig:deg4 auto2}
are indeed the profiles of the graphs of \cref{FellowsEtAlTwo_fig:d4 auto new,FellowsEtAlTwo_fig:d4 auto2 new}, respectively (implying that the subgraphs are 0-equivalent).

Indeed, \cref{FellowsEtAlTwo_rr:deg4 auto,FellowsEtAlTwo_rr:deg4 auto2} cover all cases of degree-four vertices with two edges in the neighborhood,
leaving only the cases of one and no edges in the neighborhood in order to be able to reduce to graphs of minimum degree five.

\section{Conclusion and Open Problems}\label{FellowsEtAlTwo_sec:open:problems}

In \cref{FellowsEtAlTwo_sec:small:degree}, we have discussed the \emph{barrier degree} constant $\barrierdegree$ for \textsc{Vertex Cover} kernelization: We observed that, for some $\barrierdegree \in \bN$, \textsc{Vertex Cover} cannot be kernelized to instances of minimum degree $\barrierdegree$ unless $\ETH$ fails. In terms of $\FPT$ algorithms, the equivalent concept is that of the existence of the \emph{barrier constant} $\barrierconstant > 0$ which is such that there is no algorithm for \textsc{Vertex Cover} running in time $(1 + \barrierconstant)^k\cdot n^{\cO(1)}$ modulo $\ETH$ (e.g., \cite{FellowsEtAlTwo_CFK15,FellowsEtAlTwo_DF13,FellowsEtAlTwo_Fel14}). 
So far it is only known that $\barrierconstant < 0.2738$~\cite{FellowsEtAlTwo_CKX10} and that $\barrierdegree > 3$ (\cref{FellowsEtAlTwo_sec:small:degree}, see also \cite{FellowsEtAlTwo_FS99}).
However, observe that the question of determining the concrete value of $\barrierdegree$ is much more tangible than the one of finding the value of $\barrierconstant$: Suppose one can show that a reduction rule that kernelizes degree-$d$ vertices violates $\ETH$, for some $d \in \bN$. Then one might be able to adapt the gadgets used in that proof to show an $\ETH$-violation via a reduction rule for degree $d+1$, $d+2$, $\ldots$ vertices as well. We pose: What is the exact value of $\barrierdegree$?
	
The main theme of this paper has been to gather (from hitherto unpublished sources), carefully verify, and advance research on the question: to what minimum degree $d$ can the (\emph{formidable naturally parameterized}) \textsc{Vertex Cover} problem be kernelized to kernels of minimum degree $d$, even if the exponent of the polynomial running time bound grows wildly in $d$? \ETH{} enforces a limit. 

\paragraph{Acknowledgements.} We thank Carsten Schubert and André Nichterlein for pointing out problems with \cref{FellowsEtAlTwo_reduction:cliques:co:matching:new,FellowsEtAlTwo_reduction:degree:two} as well as \cref{FellowsEtAlTwo_fig:deg4 auto,FellowsEtAlTwo_fig:deg4 auto2} in earlier versions of this work.

\end{document}